\documentclass[11pt]{article} 

\usepackage[margin=1in]{geometry} 
\usepackage[utf8]{inputenc}
\usepackage[T1]{fontenc}
\usepackage{mathpazo} 

\usepackage{authblk}
\usepackage{comment,xspace}
\usepackage[normalem]{ulem} 

\usepackage{booktabs,multicol,multirow,subcaption}

\usepackage{enumitem} 
\usepackage[dvipsnames]{xcolor} 
\definecolor{ptblue}{RGB}{15,76,129} 
\definecolor{ptemerald}{HTML}{009473} 
\definecolor{ptgray}{HTML}{939597} 

\usepackage{tikz} 

\usepackage{amsmath,amsfonts,amssymb,amsthm,bbm,bm,mathtools}

\let\OLDland\land
\renewcommand{\land}{\:\OLDland\:}
\let\OLDlor\lor
\renewcommand{\lor}{\:\OLDlor\:}
\let\OLDforall\forall
\renewcommand{\forall}{\;\OLDforall\:}
\let\OLDexists\exists
\renewcommand{\exists}{\;\OLDexists\,}

\DeclareMathOperator*{\argmax}{arg\,max}

\usepackage[ruled,linesnumbered,vlined]{algorithm2e}
\SetKw{Break}{break}

\SetCommentSty{mycommentfont}

\usepackage[square,sort]{natbib} 
\usepackage{hyperref}
\hypersetup{
linktocpage,
colorlinks=true,
citecolor=ptemerald, 
urlcolor=ptblue, 
linkcolor=Plum, 
}
\usepackage{cleveref} 

\theoremstyle{plain}
\newtheorem{theorem}{Theorem}[section]

\newtheorem{lemma}[theorem]{Lemma}

\theoremstyle{definition}
\newtheorem{definition}[theorem]{Definition}
\newtheorem{observation}[theorem]{Observation}

\theoremstyle{remark}
\newtheorem*{remark}{\upshape\bfseries Remark}

\newcommand{\cnt}{\mathrm{cnt}\xspace}
\newcommand{\exan}{ex-ante\xspace}
\newcommand{\expo}{ex-post\xspace}

\makeatletter 
\newcommand{\EF}[1]{\if\relax\detokenize\expandafter{\@firstofone#1{}}\relax EF\xspace\else EF#1\fi}
\makeatother
\newcommand{\EFOne}{\EF{1}\xspace}
\newcommand{\EFX}{\EF{X}\xspace}
\newcommand{\EFM}{\EF{M}\xspace}

\newcommand{\indivisibleGoods}{M}
\newcommand{\mixedGoods}{\indivisibleGoods \cup D}
\newcommand{\alloc}{\mathcal{A}} 
\newcommand{\bundle}{A} 
\newcommand{\x}{\mathbf{x}} 
\newcommand{\X}{\mathbf{X}}
\newcommand{\mbar}{\overline{m}} 

\newcommand{\LocalSearch}{\texttt{\textup{LocalSearch}}}
\newcommand{\MatchAndFreeze}{\textsc{Match\&Freeze}\xspace}

\title{Best-of-Both-Worlds Fair Allocation \\
of Indivisible and Mixed Goods\thanks{A preliminary version appears in \textit{Proceedings of the 20th Conference on Web and Internet Economics (WINE)}~\citep{BuLiLi24}.}}
\author[1]{Xiaolin Bu}
\author[2]{Zihao Li}
\author[3]{Shengxin Liu}
\author[4]{Xinhang Lu}
\author[1]{Biaoshuai Tao}
\affil[1]{Shanghai Jiao Tong University, \nolinkurl{{lin_bu, bstao}@sjtu.edu.cn}}
\affil[2]{Nanyang Technological University, \nolinkurl{zihao004@e.ntu.edu.sg}}
\affil[3]{Harbin Institute of Technology, Shenzhen, \nolinkurl{sxliu@hit.edu.cn}}
\affil[4]{UNSW Sydney, \nolinkurl{xinhang.lu@unsw.edu.au}}
\date{}

\begin{document}
\maketitle

\begin{abstract}
We study the problem of fairly allocating either a set of indivisible goods or a set of mixed divisible and indivisible goods (i.e., mixed goods) to agents with additive utilities, taking the best-of-both-worlds perspective of guaranteeing fairness properties both \exan and \expo.
The ex-post fairness notions considered in this paper are relaxations of envy-freeness, specifically, EFX for indivisible-goods allocation, and EFM for mixed-goods allocation.
For two agents, we show that there is a polynomial-time randomized algorithm that achieves ex-ante envy-freeness and ex-post EFX / EFM simultaneously.
For~$n$ agents with bi-valued utilities, we show there exist randomized allocations that are (i) ex-ante proportional and ex-post EFM, and (ii) ex-ante envy-free, ex-post EFX, and ex-post fractionally Pareto optimal.
\end{abstract}


\section{Introduction}

Fair division studies the problem of fairly allocating scarce resources among agents with heterogeneous preferences over the resources.
It is a fundamental problem in society, with real-world applications such as divorce settlement, rent division, course allocation, and more.
Two classic fairness notions in the literature are \emph{envy-freeness (EF)} and \emph{proportionality (PROP)}.
An allocation is said to be envy-free if each agent values her own bundle weakly better than any other bundle in the allocation, and proportional if every agent gets a bundle of value at least~$1/n$ times her value for the entire resources, where~$n$ is the number of agents.

Our focus in this paper is on the settings of allocating \emph{indivisible} goods as well as a mixture of both divisible and indivisible goods (henceforth referred to as \emph{mixed} goods) when agents have additive utilities, which have received significant attention in recent years~\citep{AmanatidisAzBi23,LiuLuSu24,Suksompong21,Suksompong25}.
Despite being desirable properties, neither envy-freeness nor proportionality can always be satisfied when (deterministically) allocating indivisible or mixed goods among the agents.

To circumvent the issue, relaxations of the notions have been proposed and studied.
For instance, when allocating indivisible goods, \emph{envy-freeness up to any good (\EFX)} (resp., \emph{envy-freeness up to one good (\EFOne)}) requires that an agent's envy towards another agent should be eliminated after the hypothetical removal of any (resp., some) good from the latter agent's bundle~\citep{Budish11,CaragiannisKuMo19}.
While the existence of \EFX allocations is only known in special cases (cf.~\Cref{sec:intro:related-work}), the weaker notion of \EFOne can always be satisfied, even for any number of agents with arbitrary monotonic utilities~\citep{LiptonMaMo04}.
With mixed goods, \citet{BeiLiLi21} proposed \emph{envy-freeness for mixed goods (\EFM)}, which generalizes both envy-freeness and \EFOne in the following sense: An agent is envy-free towards any agent whose bundle contains some divisible goods and \EFOne towards the rest.
They showed that an \EFM allocation always exists for any number of agents with additive utilities.

An alternative and common method to achieve fairness is through randomization~\citep[see, e.g.,][]{AbdulkadirogluSo98,BogomolnaiaMo01,BudishChKo13,AkbarpourNi20,GolzPePr24}.
With the power of randomization, envy-freeness and proportionality can both be easily and trivially achieved by giving all goods to a single agent uniformly at random.
The realized allocation, however, is patently unfair since all agents but one are left empty-handed.

\citet{AzizFrSh24} timely introduced the \emph{best-of-both-worlds (BoBW)} approach, which combines the two aforementioned methods with the goal of constructing a randomized allocation (i.e., a lottery over deterministic allocations) that is exactly fair \exan (before the randomness is realized) and approximately fair \expo (after the randomness is realized).
They showed that \exan envy-freeness and \expo \EFOne can be simultaneously achieved when agents have additive utilities.
In this paper, our goal is to strengthen the \expo fairness guarantee to \EFX when allocating indivisible goods, and, for the first time, extend the study of best-of-both-worlds fairness to the mixed-goods setting.

\subsection{Our Results}

As is common in the literature of fair division, we assume that agents have additive utilities when allocating either indivisible or mixed goods.

In \Cref{sec:2-agent}, we focus on two agents.
While the work of \citet{FeldmanMaNa24} showed that there always exists a randomized allocation that is \exan EF and \expo \EFX, we improve upon the result by devising a \emph{polynomial-time} algorithm to compute such a randomized allocation.\footnote{An independent work of \citet{GargSharma2024} also gives a polynomial-time algorithm to achieve the same goal by using a different approach. In addition to \exan EF and \expo \EFX, the method of \citet{GargSharma2024} can additionally achieve \expo $4/5$-approximation of the \emph{maximin share}.}
Built upon this result, we then show with mixed goods, a randomized allocation that is simultaneously \exan \EF{} and \expo \EFM can be found in polynomial time.

Next, we consider the $n$-agent cases, with \Cref{sec:n-agent-mixed,sec:n-agent-indivisible} focusing on mixed-goods and indivisible-goods settings, respectively.
In both cases, we assume that agents have \emph{bi-valued} utilities, i.e., each agent's utility for each good belongs to one of two possible values.\footnote{Binary utilities (the two possible values are~$0$ and~$1$) are special cases.}

In \Cref{sec:n-agent-mixed}, we devise a polynomial-time algorithm to compute an integral allocation sampled from a randomized allocation that is \exan proportional and \expo \EFM.
Our result on the compatibility between \EFM and \exan fairness notions adds to the growing literature revolved around \EFM when allocating mixed resources~\citep{BhaskarSrVa21,LiLiLu23,LiLiLu24,NishimuraSu23,LiLiLi24}.

In \Cref{sec:n-agent-indivisible}, we devise a polynomial-time algorithm to compute an integral allocation sampled from a randomized allocation that is \exan EF, \expo \EFX, and \expo fPO.\footnote{Ex-post \emph{fractionally Pareto optimality (fPO)}, which can be found in \Cref{def:fPO}, is a stronger notion of economic efficiency than ex-post Pareto optimality (PO).}
This generalizes multiple results known in the literature.
For instance, the compatibility between \exan EF, \expo EFX, and \expo PO was only known for binary utilities~\citep{BabaioffEzFe21,HalpernPrPs20}.\footnote{The result of \citet{BabaioffEzFe21} works for \emph{binary submodular} (also known as \emph{matroid-rank}) utilities.}
For bi-valued utilities, we only knew the compatibility of \expo notions between EFX and PO~\citep{AmanatidisBiFi21} as well as between EFX and fPO~\citep{GargMu23}.

\subsection{Related Work}
\label{sec:intro:related-work}

The fair allocation of indivisible goods has received extensive attention in the past decades~\citep{AmanatidisAzBi23,Suksompong21,Suksompong25}.
\citet{LiuLuSu24} overviewed the recent developments of mixed-goods allocations.
Below, we first discuss EFX existence results, followed by an overview of best-of-both-worlds fairness in fair division.

\paragraph{EFX Existence}
The existence of \EFX allocations remains largely open, though progress has been made in special cases where the number of agents or agents' utility functions are restricted.
An \EFX allocation always exists for two agents with general utilities~\citep{PlautRo20}, for three agents~\citep{AkramiAlCh23,ChaudhuryGaMe24}, as well as for any number of agents with restricted utilities such as being identical~\citep{PlautRo20}, binary submodular~\citep{BabaioffEzFe21} or more general~\citep{BuSoYu23}, bi-valued~\citep{AmanatidisBiFi21,GargMu23}, or when there are two types of agents and agents of the same type have identical additive utilities~\citep{Mahara23}.
There have been lines of work focusing on approximately-\EFX allocations~\citep{AkramiAlCh23,AmanatidisMaNt20}, partial allocations that are \EFX~\citep{BergerCoFe22,CaragiannisGrHu19,ChaudhuryKaMe21}, or a combination of both~\citep{ChaudhuryGaMe23-MOOR}.

\paragraph{BoBW Fairness}
In addition to the compatibility result between \exan EF and \expo \EFOne mentioned above, \citet{AzizFrSh24} showed several impossibility results regarding achieving BoBW fair and economically efficient allocations.
For agents with weighted \emph{entitlements}, \exan weighted envy-freeness (WEF) is compatible with \expo weighted transfer envy-freeness up to one good, but not compatible with any stronger \expo WEF relaxation~\citep{AzizGaMi23,HoeferScVa23}.
For agents with subadditive utilities, \exan $\frac{1}{2}$-EF, \expo $\frac{1}{2}$-\EFX and \expo \EFOne can be achieved simultaneously~\citep{FeldmanMaNa24}.
Best-of-both-worlds fairness has also been explored for fair-share-based notions like proportionality and the \emph{maximin share (MMS)} guarantee for agents with additive~\citep{BabaioffEzFe22,AkramiGaSh24} or fractionally subadditive utilities~\citep{AkramiMeSe23}.
\citet{BabaioffEzFe22} showed \exan proportionality and \expo $\frac{1}{2}$-MMS are compatible.
The ex-post MMS approximation ratio was improved in~\citep{AkramiGaSh24}, at the cost of weakening ex-ante fairness guarantees.

Slightly further afield, the BoBW paradigm has also been applied to the contexts of collective choice such as committee voting~\citep{AzizLuSu23,SuzukiVo24} and participatory budgeting~\citep{AzizLuSu24}.

\section{Preliminaries}

For each positive integer~$t$, let $[t] \coloneqq \{1, 2, \dots, t\}$.
Denote by $N = [n]$ the set of~$n$ agents to whom we allocate resources.
In this paper, we consider both indivisible-goods and mixed-goods settings.
Let~$\mixedGoods$ be the set of mixed goods, where $\indivisibleGoods = \{g_1, g_2, \dots, g_m\}$ is the set of~$m$ indivisible goods and $D = \{d_1, d_2, \dots, d_{\mbar}\}$ is the set of~$\mbar$ \emph{homogeneous} divisible goods.
When $D = \emptyset$, we are concerned with the indivisible-goods setting.

\paragraph{Allocations}

A \emph{fractional allocation} of mixed goods~$\mixedGoods$ to the agents in~$N$ is specified by a non-negative $n \times (m + \mbar)$ matrix~$\X = (X_{ig})_{i \in N, g \in \mixedGoods}$ such that for each~$g \in \mixedGoods$, $\sum_{i \in N} X_{ig} = 1$; here, $X_{ig} \in [0, 1]$ denotes the fraction of good~$g$ assigned to agent~$i$.
The $i$-th row~$X_i$ of the matrix denotes the goods allocated to agent~$i$ in the fractional allocation.
When we simply say ``an allocation'', it will mean a fractional allocation, unless otherwise clear from the context.

A fractional allocation~$\X$ is \emph{integral} if $X_{ig} \in \{0, 1\}$ for every~$i \in N$ and every indivisible good~$g \in \indivisibleGoods$.
Given an integral allocation~$\X$, denote by $\indivisibleGoods_i \coloneqq \{g \in \indivisibleGoods \mid X_{ig} = 1\}$ the set of indivisible goods assigned to agent~$i$, and let the $\mbar$-dimensional vector~$\x_i = (x_{i1}, x_{i2}, \dots, x_{i\mbar}) \coloneqq (X_{ig})_{g \in D}$ represent the divisible goods received by agent~$i$.
We will refer to $\bundle_i = (\indivisibleGoods_i, \x_i)$ as the \emph{bundle} of mixed goods allocated to agent~$i$.
The integral allocation is written as $\alloc = (\bundle_i)_{i \in N}$.
We slightly abuse notation by also using~$\alloc$ as an $n \times (m + \mbar)$-matrix representation of the integral allocation.

A \emph{randomized allocation} is a probability distribution over integral allocations and specified by a set of~$s \in \mathbb{N}$ tuples $\{(p_j, \alloc_j)\}_{j \in [s]}$, where $p_j \in [0, 1]$, $\sum_{j \in [s]} p_j = 1$, and $\alloc_j$ is an integral allocation implemented with probability~$p_j$.
A randomized allocation $\{(p_j, \alloc_j)\}_{j \in [s]}$ is called an \emph{implementation} of (or, interchangeably, \emph{implements}) a fractional allocation~$\X$ if $\X = \sum_{j \in [s]} p_j \cdot \alloc_j$.

\paragraph{Utilities}
Each agent~$i \in N$ is associated with an \emph{additive} utility function~$u_i$, where $u_i(g) \geq 0$ denotes the agent's utility for fully receiving good~$g \in \mixedGoods$.
We say that agents have \emph{bi-valued} utilities if for each~$i \in N$ and~$g \in \mixedGoods$, $u_i(g) \in \{a, b\}$, where $0 \leq a < b$ are distinct, non-negative real values.\footnote{We will later use an agent's ``large good'' (resp., ``small good'') to refer to a good of value~$b$ (resp.,~$a$) to the agent.}
When $a = 0$ and $b = 1$, agents are said to have \emph{binary} utilities.
Given an allocation~$\X$, the utility of agent~$i$ is given by $u_i(\X) = u_i(X_i) = \sum_{g \in \mixedGoods} X_{ig} \cdot u_i(g)$.

\paragraph{Solution Concepts}
We now introduce fairness notions.
The first two concepts---envy-freeness~\citep{Foley67,Varian74} and proportionality~\citep{Steinhaus48}---are defined for allocations, fractional or integral.
An allocation~$\X$ is said to satisfy \emph{envy-freeness (\EF{})} if for any pair of agents~$i, j \in N$, $u_i(X_i) \geq u_i(X_j)$, and \emph{proportionality (PROP)} if for any agent~$i \in N$, $u_i(X_i) \geq \frac{u_i(\mixedGoods)}{n}$.
Our next notions are compelling relaxations of envy-freeness when concerning integral allocations of indivisible goods.

\begin{definition}[\EFOne~\citep{LiptonMaMo04,Budish11} and \EFX~\citep{CaragiannisKuMo19,PlautRo20}]
An integral allocation $(\indivisibleGoods_i)_{i \in N}$ of indivisible goods~$\indivisibleGoods$ is said to satisfy
\begin{itemize}
\item \emph{envy-freeness up to one good (\EFOne)} if for any pair of agents~$i, j \in N$ such that $\indivisibleGoods_j \neq \emptyset$, we have $u_i(\indivisibleGoods_i) \geq u_i(\indivisibleGoods_j \setminus \{g\})$ for some~$g \in \indivisibleGoods_j$;

\item \emph{envy-freeness up to any good (\EFX)} if for any pair of agents~$i, j \in N$ such that $\indivisibleGoods_j \neq \emptyset$, we have $u_i(\indivisibleGoods_i) \geq u_i(\indivisibleGoods_j \setminus \{g\})$ for all~$g \in \indivisibleGoods_j$.
\end{itemize}
\end{definition}

With mixed goods, the following notion combines and generalizes both envy-freeness and \EFOne.

\begin{definition}[\EFM~\citep{BeiLiLi21}]
An integral allocation~$(\bundle_i)_{i \in N}$ of mixed goods~$\mixedGoods$ is said to satisfy \emph{envy-freeness for mixed goods (\EFM)} if for any pair of agents~$i, j \in N$,
\begin{itemize}
\item if agent~$j$'s bundle consists of only indivisible goods (i.e., $\bm{x}_j = \mathbf{0}$) and $M_j\neq \emptyset$, there exists some~$g \in \indivisibleGoods_j$ such that $u_i(\bundle_i) \geq u_i(\indivisibleGoods_j \setminus \{g\})$;
\item otherwise, $u_i(\bundle_i) \geq u_i(\bundle_j)$.
\end{itemize}
\end{definition}

Finally, we introduce an economic efficiency notion of importance in the context of fair division.

\begin{definition}[fPO]
\label{def:fPO}
An integral allocation $(\indivisibleGoods_i)_{i \in N}$ of indivisible goods~$\indivisibleGoods$ is said to satisfy \emph{fractionally Pareto optimal (fPO)} if there is no fractional allocation~$\mathbf{Y}$ of indivisible goods~$\indivisibleGoods$ such that $u_i(Y_i) \geq u_i(\indivisibleGoods_i)$ for all agents~$i \in N$ and $u_j(Y_j) > u_j(\indivisibleGoods_j)$ for some agent~$j \in N$.
\end{definition}

We say a randomized allocation satisfies a property \exan (resp., \expo) if the fractional allocation it implements (resp., every integral allocation in its support) satisfies the property.
This paper concerns the problem of designing randomized algorithms that simultaneously achieve desirable properties both \exan and \expo.
Our algorithms do not explicitly output the desired randomized allocation and instead \emph{sample} integral allocations from their supports.

\section{Two Agents}
\label{sec:2-agent}

In this section, we study the best-of-both-worlds fair allocation of indivisible and mixed goods for two agents.
With indivisible goods, \citet{FeldmanMaNa24} showed the existence of a randomized allocation that is \exan EF and \expo EFX.
We improve upon their result by presenting a \emph{polynomial-time} algorithm to identify such an allocation.
We then leverage this algorithm to provide a randomized allocation of mixed goods that is both \exan EF and \expo EFM in polynomial time.

An independent work of \citet{GargSharma2024} gives a fundamentally different polynomial-time algorithm to achieve the same goal, \exan EF and \expo EFX.
In addition, their algorithm can achieve the extra fairness guarantee of \expo $\frac{4}{5}$-MMS.

In Proposition~A.2 and Lemma~A.1 of \citet{FeldmanMaNa24}, the requirement for \expo EFX involves initially identifying two allocations, denoted as $\mathcal{A}^1$ and $\mathcal{A}^2$, which minimize the difference between two bundles under utility functions~$u_1$ and~$u_2$ respectively. The cut-and-choose protocol is then employed to determine the owner of each bundle in both allocations. The minimization of differences ensures that, for each agent $i\in[2]$, the loss in allocation $\mathcal{A}^i$ can be bounded by the gain in allocation $\mathcal{A}^{3-i}$, thereby ensuring \exan EF.
Their mechanism involves an NP-hard step of finding an allocation that minimizes the difference.
In contrast, our focus is solely on identifying two corresponding EFX allocations under utility functions~$u_1$ and~$u_2$, ensuring that the loss in $\mathcal{A}^i$ is bounded by the gain in $\mathcal{A}^{3-i}$ for each agent $i$.
Motivated by this, we present our algorithm, whose pseudocode can be found in \Cref{alg:efx_two_agents} in \Cref{append:twoagents}.

In \Cref{alg:efx_two_agents}, we incorporate a subroutine called \LocalSearch\ (\crefrange{alg_two_agents_localsearch_begin}{alg_two_agents_localsearch_end}), which, given an arbitrary allocation $(A, B)$ and a utility function~$u$, returns an integral EFX allocation with respect to~$u$.\footnote{While our \LocalSearch\ technique and the proof technique used by \citet{PlautRo20} to show their Theorem~4.2 share certain similarities, as we both update the current non-\EFX allocation by moving some good from a bundle of higher value to a bundle of lower value, there are differences as well, mainly because we consider different utility classes.
For instance, \citeauthor{PlautRo20} showed the \emph{leximin++ solution} is \EFX for $n$ agents with general but identical utilities.
In addition to bundles' values, the leximin++ solution also takes into account of bundles' sizes.
But since we focus on two agents with additive utilities, our \LocalSearch\ only updates allocations based on bundles' values.}
The following lemma shows this subroutine produces in polynomial time an EFX allocation with a weakly smaller utility difference between its two bundles compared to the provided allocation~$(A, B)$.
Its proof, along with all other omitted proofs, can be found in the appendices.

\begin{lemma}
\label{lem:two_agents_localsearch}
\LocalSearch$(A, B, u)$ returns in polynomial time an integral EFX allocation~$(A', B')$ under utility function~$u$ with $|u(B') - u(A')| \leq |u(B) - u(A)|$.
\end{lemma}

By leveraging the subroutine, we now describe our main procedure (\crefrange{alg_two_agents_main_begin}{alg_two_agents_main_end}).
For each agent~$i \in [2]$, we maintain an allocation~$\mathcal{A}^i$ which is an integral EFX allocation under~$u_i$.
If one of them is an envy-free allocation under the utility profile~$(u_1, u_2)$ (\crefrange{alg_two_agents_main_ef_begin}{alg_two_agents_main_ef_end}), we return this allocation directly.
Otherwise, either
\begin{itemize}
\item there exists an allocation~$\mathcal{A}^i$ which has a smaller utility difference over two bundles than that in~$\mathcal{A}^{3-i}$ under~$u_{3-i}$, we update $\mathcal{A}^{3-i}$ from $\mathcal{A}^i$ by using \LocalSearch\ (lines \ref{alg_two_agents_main_update_begin}-\ref{alg_two_agents_main_doublyefx}), or
\item we find the desired \exan EF allocation $\{(0.5,\mathcal{A}^1),(0.5,\mathcal{A}^2)\}$ (line \ref{alg_two_agents_main_end}).
\end{itemize}

\begin{theorem}
\label{thm:2agentsEFX}
In the indivisible-goods setting, \Cref{alg:efx_two_agents} returns an ex-ante EF and ex-post EFX allocation in polynomial time for two agents with additive utilities.
\end{theorem}

In the context of the mixed-goods setting, \Cref{alg:efx_two_agents} can be employed to identify an allocation that is both \exan EF and \expo EFM.
Specifically, we first merge all divisible goods in~$D$ to a single indivisible good~$d$ with an equal total utility and then execute \Cref{alg:efx_two_agents} on $M \cup \{d\}$.
If the output is a single integral allocation, it is an EF (thus EFM) allocation.
Otherwise, if there exists an agent $i \in [2]$ such that the good~$d$ is in bundle~$A^i_2$ (the bundle with the larger utility), we can transfer a fraction of~$d$ to~$A^i_1$ to reach EF.
If no such agent exists, since all divisible goods are consistently in the bundle with the smaller utility, EFM can be reduced to EF1, which is then satisfied by the \expo EFX property.

\begin{theorem}
\label{thm:2agentsEFM}
In the mixed-goods setting, an ex-ante EF and ex-post EFM allocation can be found in polynomial time for two agents with additive utilities.
\end{theorem}

We remark that the above procedure can also ensure a stronger \expo fairness property called \emph{envy-freeness up to any indivisible good for mixed goods (EFXM)}~\citep{BeiLiLi21,NishimuraSu23}, which replaces the adopted EF1 criteria by the EFX criteria when comparing to the bundle containing only indivisible goods.

\section{BoBW Fairness with Mixed Goods: Ex-Ante PROP + Ex-Post EFM}
\label{sec:n-agent-mixed}

In this section, we study the best-of-both-worlds fairness in the mixed-goods setting where agents have bi-valued utilities over the set of divisible and indivisible goods.
Our main result is the following:

\begin{theorem}
\label{thm:prop_efm}
In the mixed-goods setting where agents have bi-valued utilities, there exists an algorithm which can compute in polynomial time an integral allocation sampled from a randomized allocation that is \exan proportional and \expo \EFM.
\end{theorem}

\subsection{Technical Overview}

We have devised the following techniques specifically for BoBW fairness with mixed goods.

\paragraph{Baseline comparison.}
This technique is designed to show our randomized allocation satisfies ex-ante proportionality.
Fix any agent~$i$, we construct a \emph{baseline allocation} $\mathcal{B}^i = (B_1^i, \dots, B_n^i)$ based on~$u_i$.
We partition the outcomes of the randomized allocations to~$n$ events $\mathcal{E}_1, \dots, \mathcal{E}_n$ and show that, for each $k = 1, \dots, n$, the utility for agent~$i$ under each outcome in the event~$\mathcal{E}_k$ is at least~$u_i(B_k^i)$.
The expected utility of agent $i$ is then at least $\sum_{k=1}^n\Pr(\mathcal{E}_k)\cdot u_i(B_k^i)$.
If we can show~$\mathcal{B}^i$ satisfies
\begin{equation}
\label{eqn:baselinetechnique}
\sum_{k=1}^n\Pr(\mathcal{E}_k)\cdot u_i(B_k^i)\geq\frac1n\cdot u_i(M\cup D),
\end{equation}
then the ex-ante proportionality for agent $i$ is proved.

To apply this technique, we need to carefully design the baseline allocation $\mathcal{B}^i$ for each agent $i$ and the partition to the $n$ events such that \Cref{eqn:baselinetechnique} holds.
A natural idea is to design the partition of the event space with $\Pr(\mathcal{E}_1) = \cdots = \Pr(\mathcal{E}_n) = \frac1n$ such that \Cref{eqn:baselinetechnique} holds with equality for any baseline allocation~$\mathcal{B}^i$.
This technique is first illustrated in \Cref{sect:mlen} where $m \leq n$.
In the later part with general~$m$, more sophisticated baseline allocations are constructed.

\paragraph{Minimal unmatchable group.}
Our algorithm starts by allocating indivisible goods, and the divisible goods are allocated by a water-filling process~\citep{BeiLiLi21}.
To guarantee BoBW fairness, the indivisible goods must be allocated carefully in order to satisfy some properties that enable the application of the baseline comparison technique described above.
We have interpreted the allocation problem as a matching problem in bipartite graphs and identified a key structure, \emph{minimal unmatchable groups}, that is crucial for BoBW fairness.
Many other techniques in the bipartite graph matching problem such as \emph{augmenting paths} are also applied in our result.
See \Cref{sec:generalm} for more details.

\subsection{Preparations}

First of all, we justify that we can assume without loss of generality that there is only one divisible good, denoted as~$d$.
We also let $u_i(d) \coloneqq u_i(D) = \sum_{j = 1}^{\mbar} u_i(d_j)$ for each~$i \in N$.
Specifically, whenever we say that an~$\epsilon$ portion of~$d$ is allocated to an agent~$i$, we refer to the scenario of allocating an~$\epsilon$ portion of each divisible good $d_1, \dots, d_{\mbar}$ to the agent, i.e., $x_{i 1} = \cdots = x_{i \mbar} = \epsilon$.
In this way, the set of~$\mbar$ divisible goods can be treated as the single divisible good~$d$.
In the remainder of this section, we thus assume that there is only one divisible good $D = \{d\}$.

Given an EF1 allocation of indivisible goods, Algorithm~1 of \citet{BeiLiLi21} specifies a way to allocate the divisible goods and finally obtains an EFM allocation of the mixed goods.
We refer to this process of allocating divisible goods on top of an EF1 allocation of indivisible goods as the \emph{water-filling} process.
During the water-filling process, agents may swap their bundles; however, no bundle will be split.
This is formally stated in the following \namecref{lem:efm}.

\begin{lemma}[\citet{BeiLiLi21}]
\label{lem:efm}
Given an EF1 allocation $(M_i)_{i \in N}$ of indivisible goods~$M$, an EFM allocation $((M_{\pi(1)}, \x_1), \dots, (M_{\pi(n)}, \x_n))$ can be computed in polynomial time, where~$\pi$ is a permutation of~$[n]$.
\end{lemma}

\subsection{A Special Case where~$m \leq n$}
\label{sect:mlen}

We first consider the case~$m \leq n$, i.e., the number of indivisible goods is at most the number of agents.
This result will be used in later parts.
Note, utility functions are not required to be bi-valued.

Our algorithm is shown in \Cref{alg:prop_efm_m_le_n}.
The Round-Robin algorithm is adopted in the first step to allocate the indivisible goods where the order of the agents~$\pi$ is sampled uniformly at random.
The water-filling process is then executed to allocate divisible goods and obtain allocation $\mathcal{A}^\pi$.
The required randomized allocation we find is $\{(\frac{1}{n!}, \mathcal{A}^\pi)\}$ and is denoted by $\mathcal{R}$.

\begin{algorithm}[t]
\caption{An \exan PROP and \expo EFM randomized allocation when $m \leq n$}
\label{alg:prop_efm_m_le_n}
\DontPrintSemicolon

\KwIn{Agents~$N$, mixed goods~$\mixedGoods$, and agents' utility functions.}

$\forall i \in[n]$, $A_i \gets \emptyset$\;
Let~$\pi$ be a uniformly random permutation of~$N$.\; \label{line:prop_random_order}
\ForEach(\tcp*[f]{Round-Robin.}){agent~$i$ according to the order~$\pi$}{
	$g^* \gets \argmax_{g \in M} u_{i}(g)$\;
	$A_i \gets A_i \cup \{g^*\}$, $M \gets M \setminus \{g^*\}$\;
}
Execute the water-filling process. \tcp*{See \Cref{lem:efm}.}

\Return{Allocation~$\mathcal{A}^\pi$}\;
\end{algorithm}

\begin{theorem}
\label{thm:prop_efm_m_le_n}
In the mixed-goods setting where agents have additive utilities and $m \leq n$, \Cref{alg:prop_efm_m_le_n} computes in polynomial time an integral allocation sampled from a randomized allocation that is \exan proportional and \expo \EFM.
\end{theorem}

\begin{proof}
We can observe \Cref{alg:prop_efm_m_le_n} runs in polynomial time.
By the property of Round-Robin, the allocation after the \verb|foreach|-loop is EF1.
Then, by \Cref{lem:efm}, the final output allocation is EFM; hence, the randomized allocation is \expo EFM.

To show the randomized allocation is \exan proportional, a key observation is that, for any fixed partition $(X_1, \dots, X_n)$ of~$M \cup D$, if an agent has a probability of~$\frac1n$ to receive each bundle~$X_i$, then her expected utility is exactly~$\frac{u_i(\mixedGoods)}{n}$.
As the permutation of the agents is generated uniformly at random in \cref{line:prop_random_order}, the probability that an agent~$i$ is ranked $k$-th in the order is~$\frac1n$.
Our goal is to define a \emph{baseline allocation} $\mathcal{B}^i = (B_1^i, \dots, B_n^i)$ for each agent~$i \in N$ where the value of the bundle that agent~$i$ receives in the actual allocation when she is ranked $k$-th is no less than the $k$-th bundle in the baseline allocation (i.e., $B^i_k$) regardless of the permutation of other agents.

Given a permutation~$\pi$, we use~$A(\pi, k)$ to denote the bundle allocated to the agent ranked $k$-th in the allocation output by \Cref{alg:prop_efm_m_le_n}.
Consider an arbitrary permutation~$\pi$ where agent~$i$ is ranked $k$-th (i.e., $\pi(k) = i$), we will show that $u_i(A(\pi, k)) \geq u_i(B_k^i)$ where $\mathcal{B}^i=(B_1^i,\ldots,B_n^i)$ is the baseline allocation for agent~$i$ defined below.

\begin{definition}[Baseline allocation]
\label{def:baseline_allocation}
The \emph{baseline allocation~$\mathcal{B}^i$ for agent~$i$} is obtained by
\begin{enumerate}
\item letting agent~$i$ partition the indivisible goods~$\indivisibleGoods$ into~$n$ bundles using the Round-Robin algorithm (in particular, when $m \leq n$, the $t$-th bundle (i.e.,~$B_t^i$) contains the $t$-th preferred good for~$t \leq m$ while the $t$-th bundle is empty for~$t > m$), and

\item executing the water-filling process for the divisible goods according to agent~$i$'s utility to obtain an EFM allocation for the utility profile $(u_i, u_i, \dots, u_i)$.
\end{enumerate}
\end{definition}

The following properties of~$\mathcal{B}^i$ are straightforward.
First, each~$B_j^i$ contains at most one indivisible good (i.e., the $j$-th preferred indivisible good of agent~$i$ if $j \leq m$, and no indivisible good if $j > m$).
Second, to guarantee EFM, the bundles containing some fraction of the divisible good must have the same value under~$u_i(\cdot)$.
Let~$x$ be that value.
Third, due to EFM, $u_i(B_j^i) \geq x$ for each~$j \in [n]$.

Next, we show that $u_i(A(\pi,k))\geq u_i(B_k^i)$ for any permutation~$\pi$ with $\pi(k) = i$.
It holds trivially when $B_k^i = \emptyset$, or when $B_k^i$ only contains one indivisible good as agent~$i$ will receive at least her $k$-th preferred good during the Round-Robin phase of \Cref{alg:prop_efm_m_le_n}.
When~$B_k^i$ contains divisible goods, we need to show that $u_i(A(\pi, k)) \geq x$ (recall that~$x$ denotes the value of the bundles with divisible goods in~$\mathcal{B}^i$).
Suppose this is not the case, and $u_i(A_i) < x$ in the actual output allocation $\alloc = (A_1, \dots, A_n)$ where~$A_i = A(\pi, k)$.
Let~$\mathcal{I}$ be the set of indices of bundles in $\{A_1, \dots, A_n\}$, each of which is envied by agent~$i$.
It is worth noting that due to that allocation~$\alloc$ satisfies \EFM, those bundles do not contain any divisible good.
It is also easy to see $x > u_i(A_i) \geq u_i(A_j)$ for all~$j \notin \mathcal{I}$.
Hence, we reach the following contradiction:
\begin{equation}
\label{eqn:prop_efm_m_le_n}
u_i(M \cup D) = \sum_{j \in \mathcal{I}} u_i(A_j) + \sum_{j \notin \mathcal{I}} u_i(A_j) < \sum_{j = 1}^{|\mathcal{I}|} u_i(B_j^i) + (n - |\mathcal{I}|) \cdot x \leq \sum_{j = 1}^n u_i(B_j^i) = u_i(M \cup D).
\end{equation}

Finally, we have
\[
u_i(\mathcal{R}) = \sum_{k = 1}^n \sum_{\pi : \pi(k)=i} \frac{u_i(A(\pi,k))}{n!} \geq \sum_{k = 1}^n \sum_{\pi : \pi(k) = i} \frac{u_i(B_k^i)}{n!} = \sum_{k = 1}^n \frac{u_i(B^i_k)}{n} = \frac{u_i(\mixedGoods)}{n},
\]
which implies the allocation satisfying \exan proportionality.
\end{proof}

Note that the above algorithm can hardly be generalized to the case with more than~$n$ indivisible goods.
Suppose we define the baseline allocation for each agent in the same way and $B^i_k$ contains divisible goods, then \Cref{eqn:prop_efm_m_le_n} may fail.
In more detail, bundle~$A_j$, where $j \in \mathcal{I}$ may contain more than one indivisible good.
Hence, it may be that $\sum_{j \in \mathcal{I}} u_i(A_j) > \sum_{j = 1}^{|\mathcal{I}|} u_i(B_j^i)$.
As a result, \Cref{eqn:prop_efm_m_le_n} does not necessarily hold, leading to the possibility of $u_i(A(\pi, k)) = u_i(A_i) < u_i(B^i_k)$.

\subsection{General~$m$ with Bi-Valued Utilities}
\label{sec:generalm}
We now proceed to the general case with arbitrary numbers of goods.
We only provide the high-level ideas here, and the details are deferred to \Cref{append:nagentmixed}.

Our algorithm allocates the indivisible goods iteratively and then allocates the divisible goods by the water-filling process.
At each iteration, we attempt to allocate each agent a large good of value~$b$.
If we construct a bipartite graph $G=(U,V,E)$ where $U$ denotes the set of agents, $V$ denotes the set of indivisible goods, and an edge represents that an agent has value~$b$ to a good, then we can find a maximum matching in~$G$ at each iteration.
If a matching of size~$n$ (i.e., a perfect matching) is found, we allocate each agent a ``large good'' and move on to the next iteration.

At some iteration, a perfect matching may no longer exist.
We identify agents $Z\subseteq U$ such that
\begin{enumerate}
\item $Z$ cannot be fully matched to large goods; i.e., the Hall's condition fails for $Z$: $|\Gamma(Z)|<|Z|$, and

\item the remaining agents $U \setminus Z$ can be fully matched; in addition, for each indivisible good that has value~$b$ to an agent in~$U \setminus Z$, it has value~$a$ to any agent in~$Z$.
\end{enumerate}

Intuitively speaking, agents in~$Z$ are ``mostly finished'' with their large goods, while a perfect matching may continue to exist in the future iterations for agents in~$U \setminus Z$.
We call a minimal set of agents~$Z$ satisfying the above requirements a \emph{minimal unmatchable group}.
Its precise definition and the algorithm to find such a group are presented in \Cref{append.nagentmixed1}.

At each iteration, agents~$Z$ and indivisible goods~$\Gamma(Z)$ are temporarily removed.
We find a perfect matching for agents in~$U \setminus Z$, and allocate a ``small good'' to each agent in~$Z$.
In future iterations, we recursively handle~$U \setminus Z$.
As a result, more and more agents (and the corresponding ``neighboring goods'') are removed.
For agents that remain, we find a perfect matching and correspondingly allocate each of them a ``large good''; for agents that are removed, we allocate each of them a ``small good''.
Thus, in each iteration, exactly~$n$ goods are allocated and each agent receives exactly one good.
We stop when the number of the remaining indivisible goods is less than~$n$.

Now, the number of unallocated indivisible goods is at most~$2n-2$, including at most~$n-1$ goods that are temporarily removed by the algorithm (those that are neighbors of removed agents) and at most~$n-1$ goods that remain after the last iteration.
\emph{It is crucial that the current partial allocation satisfies envy-freeness} (which can be easily verified for each iteration based on the property of minimal unmatchable groups).
This allows us to reduce our problem to the case with at most~$2n-2$ indivisible goods, as we can fix the allocation of the indivisible goods allocated by the above procedure.
In \Cref{append.nagentmixed2}, we formally describe this procedure and the reduction to $m \leq 2n-2$.

The most technical part comes to the handling of the case with~$m \leq 2n-2$, which is described in \Cref{append.nagentmixed3}.
The case with $m \leq n$ has already been handled in \Cref{sect:mlen}.
The case where $n < m \leq 2n-2$ is handled with a similar idea of comparing to a baseline allocation but requiring a substantial amount of additional effort.
We again find a minimal unmatchable group~$Z$ of agents, and agents in~$Z$ and in~$U\setminus Z$ are handled separately in the ``last two rounds''.
Sample a permutation~$\pi$ uniformly at random.
In the first round, agents with ``higher priority'' in~$Z$ receive ``the last large good'', while all agents in $U \setminus Z$ receive large goods.
In the second round, we allocate agents with ``higher priority'' one good in the Round-Robin way.
Finally, the water-filling process is applied to allocate the divisible good.
The analysis for \exan proportionality is the most tricky part.
Different baseline allocations are used for agents in~$Z$ and agents in~$U \setminus Z$.

\section{BoBW Fairness with Indivisible Goods: \\
Ex-Ante EF + Ex-Post EFX + Ex-Post fPO}
\label{sec:n-agent-indivisible}

\begin{algorithm}[t]
\caption{Bi-valued indivisible goods: Ex-ante EF, \expo EFX, and \expo fPO}
\label{alg:efx_fpo}
\DontPrintSemicolon

\KwIn{Agents~$N$, indivisible goods~$M$, and agents' utility functions.}

$\forall i \in [n]$, $A_i \gets \emptyset$, $c_i \gets 0$\; \label{alg_efx_fpo_init_begin}
$t \gets 1$, $\cnt \gets 0$\;
Uniformly generate a random permutation~$\pi$ of the agents.\; \label{alg_efx_fpo_init_end}

\While{$|M| \geq [\#~\textup{unfrozen agents}] + \cnt$}{ \label{alg_efx_fpo_whileloop1_begin}
	Find the minimal unmatchable group~$Z_t$ among~$N$ and~$M$ according to \Cref{alg:find_unmatchable_group}.\; \label{alg_efx_fpo_zt_begin}

	\While{$|Z_t| > 0$ and there exists~$i \in N \setminus Z_t$ and~$g \in A_i$ s.t.\ $u_j(g) = b$ for an agent~$j \in Z_t$}{ \label{alg_efx_fpo_whileloop2_begin}
		Let~$g' \in M$ be the good matched to agent~$i$ in a perfect matching between $N \setminus Z_t$ and $M \setminus \Gamma(Z_t)$.\;
		$M \gets M \cup \{g\} \setminus \{g'\}$, $A_i \gets A_i \cup \{g'\} \setminus \{g\}$\;
		Find the minimal unmatchable group~$Z_t$ among~$N$ and~$M$ via \Cref{alg:find_unmatchable_group}.\; \label{alg_efx_fpo_whileloop2_end}
	}

	\eIf{$|Z_t| = 0$}{ \label{alg_efx_fpo_case1_begin}
		Find a perfect matching between~$N$ and~$M$. For each agent~$i \in N$, add her matched good to~$A_i$. Remove the matched goods from~$M$.\; \label{alg_efx_fpo_case1_end}
	}{
		\ForEach{$i \in Z_t$ in the reversed order of~$\pi$}{ \label{alg_efx_fpo_case2_begin}
			\uIf{$i$ can find an augmenting path to a good in~$\Gamma(Z_t)$}{ \label{alg_efx_fpo_case2_alloc}
				Update~$\mathcal{A}^\pi$ according to this path.\;
				Freeze~$i$ for the next $\lfloor b/a -1 \rfloor$ rounds. \;\label{alg_efx_fpo_case2_freeze}
			}\lElse{
				Set~$i$ as quiet. \label{alg_efx_fpo_case2_quiet}
			}
		}
		Match the corresponding good to each agent in~$N \setminus Z_t$ in a perfect matching between~$N \setminus Z_t$ and~$M \setminus \Gamma(Z_t)$.\; \label{alg_efx_fpo_case2_perfect_begin}
		$N \gets N \setminus Z_t$, $M \gets M \setminus \Gamma(Z_t)$, $t \gets t+1$\; \label{alg_efx_fpo_case2_end}
	}
	$\cnt \gets \cnt + [\#~\textup{quiet agents}]$\; \label{alg_efx_fpo_updatec_begin}
	$c_i \gets c_i + 1$ for each quiet agent~$i$.\; \label{alg_efx_fpo_whileloop1_end}
}

Let $M' \gets M \cup \bigcup_{i \in N} A_i$ and match each~$g \in A_i$ for some~$i \in N$ to a copy of agent~$i$ initially.\; \label{alg_efx_fpo_rem_begin}

\ForEach{agent $i$ in the first $|M|-\cnt$ of the order $\pi$ after removing the frozen agents}{
	\leIf{$i$ can find an augmenting path to a good in~$M'$}{ \label{alg_efx_fpo_rem_aug}
		Update~$\mathcal{A}^\pi$ according to this augmenting path.\;
	}{
		$c_i\leftarrow c_i+1$ \label{alg_efx_fpo_rem_end}
	}
}

Allocate the remaining goods arbitrarily based on~$c_i$.\; \label{alg_efx_fpo_clean}

\Return{Allocation~$\mathcal{A}^\pi$}
\end{algorithm}

In this section, we focus on the indivisible-goods setting and investigate the best-of-both-world guarantee for both fairness and efficiency when the agents have bi-valued utilities.
Here, we will assume $a > 0$ for a clearer demonstration.
For the case of $a = 0$, it can be reduced to the binary setting and has been addressed~\citep{BabaioffEzFe21}.
Our main result is the following:

\begin{theorem}
\label{thm:efx_fpo}
In the indivisible-goods setting where agents have bi-valued utilities, \Cref{alg:efx_fpo} computes in polynomial time an integral allocation sampled from a randomized allocation that is ex-ante EF, ex-post EFX and ex-post fPO.
\end{theorem}

Theorem~5 of \citet{AzizFrSh24} presented an instance with two goods and two agents with additive utilities to show that no randomized allocation is simultaneously \exan EF, \expo EF1, and \expo fPO.
Their impossibility result continues to hold when replacing \exan EF by \exan PROP.
It is worth noting that the instance consists of three possible values for the goods, which can be considered in the settings allowing tri-valued utilities or \emph{personalized} bi-valued utilities.
Their impossibility result indicates that the assumption of bi-valued utilities is \emph{necessary} in order to achieve best-of-both-world fairness \emph{and} economic efficiency.
For tri-valued or personalized bi-valued utilities, \citeauthor{AzizFrSh24}'s impossibility result holds for the weaker notions of ex-ante PROP and ex-post EF1, with two agents and two goods.
For bi-valued utilities, our \Cref{thm:efx_fpo} indicates that the best-of-both-worlds fairness can hold for the stronger notions of ex-ante EF and ex-post EFX, with an arbitrary number of agents and goods.
Our result shows a sharp contrast between bi-valued and tri-valued utilities, and thereby completes the whole picture.

\subsection{Technique Discussion and Detailed Descriptions of Algorithm~\ref{alg:efx_fpo}}
\label{sec:ind:alg-description}

For bi-valued utilities, \citet{GargMu23} devised a polynomial-time algorithm to find an EFX and fPO allocation.
Their algorithm updates the allocation through consecutive transfers of goods along a path, a process seemingly challenging to implement for achieving an \exan fairness guarantee.
We instead draw inspiration from the \MatchAndFreeze algorithm of \citet{AmanatidisBiFi21}, which computes an \EFX allocation for bi-valued utilities in polynomial time.
It is worth noting that the \MatchAndFreeze need not return a PO allocation~\citep[Appendix~A.2]{GargMu23}.
In what follows, we present the first polynomial-time algorithm which can return an \exan EF, \expo EFX, and \expo fPO allocation when agents have bi-valued utilities.
When describing our \Cref{alg:efx_fpo}, we use the term ``round'' interchangeably with each execution of the loop in \cref{alg_efx_fpo_whileloop1_begin}.

We first introduce three possible states for the agents: \emph{active}, \emph{quiet} and \emph{frozen}.
Initially, all agents are \emph{active}.
When executing \crefrange{alg_efx_fpo_case2_begin}{alg_efx_fpo_case2_end}, if an agent~$i \in Z_t$ is allocated a large good in \cref{alg_efx_fpo_case2_alloc}, she will become \emph{frozen} for the next $\lfloor b/a-1 \rfloor$ rounds, indicating that she will not receive any good in these $\lfloor b/a-1 \rfloor$ rounds.
After those rounds, she will become \emph{quiet}.
If $i \in Z_t$ does not receive a large good at this round, she will directly become \emph{quiet} in \cref{alg_efx_fpo_case2_quiet}.
Intuitively, during the execution of the loop in \crefrange{alg_efx_fpo_whileloop1_begin}{alg_efx_fpo_whileloop1_end}, an agent~$i$ can only accept small goods after she turns quiet, and the count for these goods is recorded in~$c_i$.
We also refer to active and quiet agents as \emph{unfrozen} agents.

We now describe \Cref{alg:efx_fpo}.
Lines~\ref{alg_efx_fpo_init_begin}-\ref{alg_efx_fpo_init_end} initialize the algorithm, where~$c_i$ records the number of reserved small goods for agent~$i$, crucial for the market clear in \cref{alg_efx_fpo_clean}.
The counter~$\cnt$ records the total number of reserved goods, precisely the sum of all~$c_i$'s during the loop (lines~\ref{alg_efx_fpo_whileloop1_begin}-\ref{alg_efx_fpo_whileloop1_end}).
A permutation~$\pi$ is generated uniformly at random for subsequent allocations.

The main algorithm proceeds through a multi-round procedure in lines~\ref{alg_efx_fpo_whileloop1_begin}-\ref{alg_efx_fpo_whileloop1_end}.
In each round, we first check whether the remaining goods in~$M$ are sufficient to match one to each unfrozen agent in \cref{alg_efx_fpo_whileloop1_begin}.
If this condition is met, we attempt to find a specific minimal unmatchable group~$Z_t$ such that all agents in~$Z_t$ have value~$a$ over each good matched to agents in~$N \setminus Z_t$ in lines~\ref{alg_efx_fpo_zt_begin}-\ref{alg_efx_fpo_whileloop2_end}.\footnote{The relevant concepts can be found in \Cref{append.nagentmixed1}.}
This is to ensure the efficiency guarantee by allocating more large goods intuitively.

Given~$Z_t$, if~$|Z_t| = 0$, indicating that there exists a perfect matching between all agents in~$N$ and the current set of goods in~$M$, lines~\ref{alg_efx_fpo_case1_begin}-\ref{alg_efx_fpo_case1_end} are executed to add this matching directly.
When~$Z_t$ is non-empty, we first find a maximum matching between~$Z_t$ and~$\Gamma(Z_t)$ using the augmenting path technique in the reversed order of~$\pi$ (lines~\ref{alg_efx_fpo_case2_begin}-\ref{alg_efx_fpo_case2_quiet}).
All goods in~$\Gamma(Z_t)$ can be allocated to some agent from the construction of~$Z_t$ (see \Cref{alg:find_unmatchable_group}).
The reversed order in \cref{alg_efx_fpo_case2_begin} ensures \expo EFX after the final allocation step in lines~\ref{alg_efx_fpo_rem_begin}-\ref{alg_efx_fpo_rem_end}.
We then find a perfect matching between $N \setminus Z_t$ and $M \setminus \Gamma(Z_t)$ in lines~\ref{alg_efx_fpo_case2_perfect_begin}-\ref{alg_efx_fpo_case2_end} and match these goods accordingly.
At the end of each round, we maintain~$c_i$ for each quiet agent~$i$ and the counter~$\cnt$, to allocate a good to those who have not received a good at this round \emph{virtually} (lines~\ref{alg_efx_fpo_updatec_begin}-\ref{alg_efx_fpo_whileloop1_end}).

If there are not enough goods to allocate one to each unfrozen agent, the \verb|while|-loop in \cref{alg_efx_fpo_whileloop1_begin} is terminated and the final allocation step in lines~\ref{alg_efx_fpo_rem_begin}-\ref{alg_efx_fpo_rem_end} is executed.
At this step, a large good is allocated to the first $|M| - \cnt$ unfrozen agents under~$\pi$ as far as possible.
\Cref{alg_efx_fpo_rem_begin} can adjust the previous allocation for the remaining agents in~$N$ for a better efficiency guarantee.
\Cref{alg_efx_fpo_rem_end} updates the number of reserved small goods for an agent if she cannot be matched to any large good.
\Cref{alg_efx_fpo_clean} clears the market based on the reserved number $c_i$ for each agent $i$, which allocates these goods \emph{actually}.
We can then obtain the desired randomized allocation.

Throughout the algorithm, the randomization of our algorithm (the choice of~$\pi$) only influences the allocation in lines~\ref{alg_efx_fpo_case2_begin}-\ref{alg_efx_fpo_case2_end} and~\ref{alg_efx_fpo_rem_begin}-\ref{alg_efx_fpo_rem_end}, which would not affect the assignment of agents to groups~$Z_t$ (resp., goods to~$\Gamma(Z_t)$).
This observation will be helpful for the analysis of the \exan property.

\subsection{Fairness}

Before presenting the results, we first provide some properties of the allocation returned by \Cref{alg:efx_fpo}.
Let~$N^t$ and~$M^t$ be the corresponding set~$N$ and~$M$ at the beginning of round~$t$.

\begin{observation}
\label{obs:efx_fpo}
For each agent $i \in N$,
\begin{itemize}
\item she receives exactly one good at each unfrozen round;

\item if she belongs to some~$Z_t$ which is removed from~$N$ in \cref{alg_efx_fpo_case2_end} (assuming this corresponds to round~$r_i$), she values each good allocated to her before round~$r_i$ at~$b$ and each good in~$A_j$ for~$j \in N^{r_i} \setminus Z_t$ and $M^{r_i} \setminus \Gamma(Z_t)$ at~$a$;

\item if she does not belong to any~$Z_t$, she values each good allocated to her before the termination of the \verb|while|-loop in \cref{alg_efx_fpo_whileloop1_begin} at~$b$.
\end{itemize}
\end{observation}

According to these observations, we first show that each output integral allocation by our algorithm is EFX.
To prove this, we adopt the induction to maintain the following property after the matching at each round: for each agent~$i$, this agent will not envy other agents~$j$ except for the case that agents~$i$ and~$j$ are in the same~$Z_t$ and agent~$i$ has never been frozen while agent~$j$ has been frozen.
Under this case, agent~$i$ does not envy~$j$ using the EFX criteria.
Together with the analysis for the final allocation step, we conclude the lemma.

\begin{lemma}
\label{lem:efx_fpo_expost_efx}
Every integral allocation returned by \Cref{alg:efx_fpo} is EFX.
\end{lemma}

We then come to the \exan fairness guarantee of our algorithm.
The idea is to prove any envy from an agent~$i$ to another agent~$j$ under~$\pi$ can be eliminated by the gain under another permutation~$\pi'$ which just exchanges~$i$ and~$j$ in~$\pi$.

\begin{lemma}
\label{lem:efx_fpo_exante_ef}
\Cref{alg:efx_fpo} returns an ex-ante EF allocation.
\end{lemma}

\subsection{Efficiency}

We leverage the Fisher market (relevant concepts in Appendix) to provide our efficiency guarantee.

\begin{lemma}
\label{lem:efx_fpo_expost_fpo}
Every integral allocation returned by Algorithm~\ref{alg:efx_fpo} is fPO.
\end{lemma}

\begin{proof}
    To leverage Theorem~\ref{thm:fisher_market} for proving this, our goal is to find a proper price vector $\bm{p}$ for each realized integral allocation $\mathcal{A}$ such that $(\mathcal{A},\bm{p})$ is a market equilibrium of a Fisher market.
    Here, we can directly set the budget $e_i$ for each agent $i\in N$ as $\sum_{g\in A_i}p_g$ to ensure the market clear. It suffices to find $\bm{p}$ and show that each agent $i\in N$ only receives goods in $\rm{MBB}_i$.

    Denote the set of agents not in any previous $Z_t$ by $N^r$.
    We assume $m\ge n$ here and the analysis for the corner case when $m<n$ is deferred to Appendix.
    Define the price vector $\bf{p}$ as follows:
    \begin{itemize}
        \item for each agent~$i$ in some $Z_t$, we set $p_g=u_i(g)$ for $g\in A_i$;
        \item for each agent $i\in N^r$ who receives both large and small goods, set $p_g=u_i(g)$ for each $g\in A_i$.
    \end{itemize}
    For the prices of the remaining goods, we will set them iteratively: if there exists one agent $i$ without pricing her goods such that some agent who has priced her goods before values some good in $A_i$ at $b$, we set $p_g=u_i(g)$ for each $g\in A_i$.
    If there is no further pricing that can be made, we set the price as $a$ for each remaining good and we call each agent who owns these goods \emph{low-price agent}.
    We then need to show each agent $i\in N$ only receives goods in $\rm{MBB}_i$.

    For each agent $i$ in some $Z_t$, the ratio $u_i(g)/p_g$ for each good $g\in A_i$ is exactly $1$.
    For each good $g$ which is allocated to some agent $j$ in a $Z_{t'\le t}$, if $p_g=u_j(g)=a$, this good is \emph{actually} allocated at line \ref{alg_efx_fpo_clean} and $g$ is in $M^{r_i}\setminus\Gamma(Z_t)$, which is valued $a$ by $i$ from the second term in Observation~\ref{obs:efx_fpo}, whose ratio $u_i(g)/p_g$ is at most $1$.
    For each good $g$ which is allocated to some agent $j$ in $N^{r_i}\setminus Z_t$, we also have $u_i(g)=a$ from Observation~\ref{obs:efx_fpo}.
    Thus, each agent $i$ in some $Z_t$ only receives goods in $\rm{MBB}_i$.

    For each low-price agent, from the third term in Observation~\ref{obs:efx_fpo}, all goods allocated to her in the multi-round procedure must be large.
    Because this agent is not priced before the end of the iteration, she receives no small goods.
    Since the goods she owns are priced at $a$ which is the lowest price we set and all goods in her bundle must be valued at $b$ under her valuation, all goods in her bundle must be in $\rm{MBB}_i$.

    We then come to some agent $i\in N^r$ who sets $p_g=u_i(g)$ for each $g\in A_i$. The ratio $u_i(g)/p_g$ for each good $g\in A_i$ is exactly $1$.
    If she values a good $g$ in some $A_j,j\in N$ at $b$ (thus $j$ cannot be a low-price agent otherwise it will be priced before the end from the iteration) and $u_j(g)=p_g=a$, we can backward the iteration from agent $i$ and finally achieve an agent $i'\in N^r$ who receives both large and small goods from her perspective. This backward path along with the good presented in the condition of the iteration is a feasible augmenting path at line \ref{alg_efx_fpo_rem_aug} and agent $i'$ would not receive a small good, which leads to a contradiction. This completes the proof for the case when $m\ge n$.
\end{proof}

We are now ready to prove our main result in this section.

\begin{proof}[Proof of \Cref{thm:efx_fpo}]
    From Lemmas~\ref{lem:efx_fpo_expost_efx},\ref{lem:efx_fpo_exante_ef} and \ref{lem:efx_fpo_expost_fpo}, it suffices to show Algorithm~\ref{alg:efx_fpo} can terminate in polynomial time.
    From the description of our algorithm, since $|M|$ keeps decreasing which bounds the total number of rounds by $m$ and the only thing we need is to show  lines~\ref{alg_efx_fpo_whileloop2_begin}-\ref{alg_efx_fpo_whileloop2_end} can be terminated in polynomial time.
    From the construction of $Z_t$ according to Algorithm~\ref{alg:find_unmatchable_group}, if we keep the perfect matching between $N\setminus Z_t$ and $M\setminus \Gamma(Z_t)$ by just replacing the good $g'$ by $g$, one edge from $j\in Z_t$ to $g'$ is added, this leads to one of the following cases: either (1) one augmenting path can be performed and the size of the maximum matching between $N$ and $M$ is increased by $1$, or (2) there is no additional augmenting path but $|\Gamma(Z_t)|$ is increased by at least $1$.
    After at most $n$ times of the case (2) occurs, case (1) will occur.
    From the total number of the occurrences of (1) is at most $n$, lines~\ref{alg_efx_fpo_whileloop2_begin}-\ref{alg_efx_fpo_whileloop2_end} can be terminated in $O(n^2)$ loops.
\end{proof}

\section{Conclusion}

In this paper, we have studied the best-of-both-worlds fairness for indivisible-goods and mixed-goods allocations.
With indivisible goods, we provided polynomial-time algorithms that achieve \exan EF and \expo EFX allocation for two agents and \exan EF, \expo EFX, and \expo fPO for $n$ agents with bi-valued utilities.
With mixed goods, we showed polynomial-time algorithms that achieve \exan EF and \expo EFM allocation for two agents and \exan PROP and \expo EFM for $n$ agents with bi-valued utilities.

In future research, it would be interesting to further strengthen the results in this paper, the most intriguing of which is perhaps the (in)compatibility between \exan \EF{} and \expo \EFM in the mixed-goods setting.
For mixed goods, another interesting direction is to extend BoBW-fair study to maximin share (MMS) guarantee~\citep{BeiLilu21}.

\section*{Acknowledgements}

The authors would like to thank Jiaxin Song for his great suggestions on this paper.

This work was partially supported by ARC Laureate Project FL200100204 on ``Trustworthy AI'', by the National Natural Science Foundation of China (Grant No.\ 62102117), by the Shenzhen Science and Technology Program (Grant No.\ GXWD20231129111306002), by the Guangdong Basic and Applied Basic Research Foundation (Grant No.\ 2023A1515011188), and by the National Natural Science Foundation of China (No.\ 62102252).

\bibliographystyle{plainnat}
\bibliography{bibliography}

\clearpage
\appendix

\section{Omitted Proofs in Section~\ref{sec:2-agent}}
\label{append:twoagents}

\begin{algorithm}[t]
\caption{An \exan EF and \expo EFX randomized allocation for two agents}
\label{alg:efx_two_agents}
\DontPrintSemicolon

\SetKwFunction{FMain}{\LocalSearch}
\KwIn{Agents~$N = [2]$ and indivisible goods~$M$}

For each~$i \in [2]$, $\mathcal{A}^i \gets$ \FMain{$\emptyset$, $M$, $u_i$}\; \label{alg_two_agents_main_begin}
\If{$\exists i \in [2]$ such that $u_i(A^i_1) = u_i(A^i_2)$ or $u_{3-i}(A^i_1) \geq u_{3-i}(A^i_2)$}{ \label{alg_two_agents_main_ef_begin}
	\Return{$\{(1, \mathcal{A}^i)\}$ where agent~$3-i$ picks her preferred bundle first} \label{alg_two_agents_main_ef_end}
}

\While{$\exists i \in[2]$ such that $u_{3-i}(A^i_2) - u_{3-i}(A^i_1) < u_{3-i}(A^{3-i}_2) - u_{3-i}(A^{3-i}_1)$}{ \label{alg_two_agents_main_update_begin}
	$\mathcal{A}^{3-i} \gets$ \FMain{$A^i_1$, $A^i_2$, $u_{3-i}$}\; \label{alg_two_agents_main_update_ls}
	\If{$\exists j \in [2]$ such that $u_j(A^j_1) = u_j(A^j_2)$ or $u_{3-j}(A^j_1) \geq u_{3-j}(A^j_2)$}{ \label{alg_two_agents_main_ef2}
		\Return{$\{(1, \mathcal{A}^j)\}$ where agent~$3-j$ picks her preferred bundle first} \label{alg_two_agents_main_update_end}
	}
	\If{$\mathcal{A}^{3-i}$ is also EFX under~$u_i$}{ \label{alg_two_agents_main_doublyefx_cond}
		$\mathcal{A}^i \gets \mathcal{A}^{3-i}$\;
		\Return{$\{(0.5, \mathcal{A}^1), (0.5, \mathcal{A}^2)\}$ where agent~$3-i$ picks her preferred bundle first in the realized allocation~$\mathcal{A}^i$ for each~$i \in [2]$}  \label{alg_two_agents_main_doublyefx}}
}

\Return{$\{(0.5, \mathcal{A}^1), (0.5, \mathcal{A}^2)\}$ where agent~$3-i$ picks her preferred bundle first in the realized allocation~$\mathcal{A}^i$ for each~$i \in [2]$} \label{alg_two_agents_main_end}

\BlankLine
\hrule
\BlankLine

\SetKwProg{Fn}{Function}{:}{\KwRet}
\Fn{\FMain{$A$, $B$, $u$}}{ \label{alg_two_agents_localsearch_begin}
	\lIf{$u(A) > u(B)$}{
		Swap~$A$ and~$B$. \label{alg_two_agents_localsearch_swap1}
	}
	\While{$\exists g \in B$ such that $u(A \cup \{g\}) < u(B)$}{ \label{alg_two_agents_localsearch_con1}
            $g \gets \argmax_{g'\in B\text{ and } u(A \cup \{g'\}) < u(B)}u(g')$\; \label{alg_two_agents_localsearch_select}
		$A \gets A \cup \{g\}$, $B \gets B \setminus \{g\}$\; \label{alg_two_agents_localsearch_move}
		\lIf{$u(A) > u(B)$}{ \label{alg_two_agents_localsearch_con2}
			Swap~$A$ and~$B$. \label{alg_two_agents_localsearch_swap}
		}
	}
	\KwRet ($A$, $B$)\; \label{alg_two_agents_localsearch_end}
}
\end{algorithm}

\subsection{Proof of Lemma~\ref{lem:two_agents_localsearch}}
\label{appendix:two_agents_localsearch}

The termination condition ensures that in the output allocation $(A',B')$, we have $u(A') \leq u(B')$ and for every good~$g \in B'$ we have $u(A' \cup \{g\}) \geq u(B')$, which meets the EFX condition.
It suffices to show this subroutine can terminate in polynomial time.

Without loss of generality, we assume $u(A) \leq u(B)$ in the initial allocation~$(A, B)$.
When executing \cref{alg_two_agents_localsearch_move}, we choose the item $g\in B$ with the highest value according to $u$ such that $u(A \cup \{g\}) < u(B)$ and move it from $B$ to $A$.
Before the next swap step, both the utility difference between the two bundles and the size of~$B$ are decreasing.
If the swap step (\cref{alg_two_agents_localsearch_swap}) is never executed, the number of steps is upper bounded by~$m$.
If one swap step (\cref{alg_two_agents_localsearch_swap}) is executed, assume the last item moved from $B$ to $A$ is $g$.
This indicates that, before moving the item $g$ from bundle $B$ to bundle $A$, we have $u(B) - u(A) > u(g)$ and $u(A) + u(g) > u(B) - u(g)$ from the conditions in \cref{alg_two_agents_localsearch_con1,alg_two_agents_localsearch_con2}, so the utility difference between the two bundles decreases from a value larger than~$u(g)$ to a value less than~$u(g)$ after moving the item $g$.
Thus, during the whole subroutine, the utility difference keeps decreasing, and an item $g'$ will not be further moved between the two bundles if $u(g')\ge u(g)$.
The values of the items that have been moved are at least $u(g)$ (due to \cref{alg_two_agents_localsearch_select}), so an item can be moved between the two bundles at most once.
The overall time complexity is bounded by~$O(m\log m)$, where we may first sort the items according to $u$, and then perform the above operations in linear time.

\subsection{Proof of Theorem~\ref{thm:2agentsEFX}}

If the allocation is returned in \cref{alg_two_agents_main_doublyefx}, this is \expo EFX from \Cref{lem:two_agents_localsearch} and the condition in \cref{alg_two_agents_main_doublyefx_cond}, and trivially \exan EF since we only exchange two bundles in these two realizations.
If the allocation is returned in \cref{alg_two_agents_main_ef_end,alg_two_agents_main_update_end}, either the agent~$i$ in \cref{alg_two_agents_main_ef_begin} or the agent~$j$ in \cref{alg_two_agents_main_ef2} treats two bundles equally or the other agent will choose the bundle which is less valued by this agent.
In both cases, this is exactly EF, in both \exan and \expo senses.
If the allocation is returned in \cref{alg_two_agents_main_end} (which also means the conditions in \cref{alg_two_agents_main_ef_begin,alg_two_agents_main_ef2} fail), from \Cref{lem:two_agents_localsearch}, both~$\mathcal{A}^1$ and~$\mathcal{A}^2$ are EFX if for each~$i \in [2]$, agent~$3-i$ picks her preferred bundle first in~$\mathcal{A}^i$.
Because condition in \cref{alg_two_agents_main_update_begin} is violated, for each agent~$i \in [2]$, we have $u_i(A^i_2) - u_i(A^i_1) \leq u_i(A^{3-i}_2) - u_i(A^{3-i}_1)$, so the loss of $u_i(A^i_2) - u_i(A^i_1)$ for agent~$i$ under~$\mathcal{A}^i$ is upper bounded by the gain of $u_i(A^{3-i}_2) - u_i(A^{3-i}_1)$ under~$\mathcal{A}^{3-i}$.
Thus, this randomized allocation is \exan EF and \expo EFX.
It suffices to show that this algorithm can terminate in polynomial steps, where we adopt a similar analysis.

From the condition at \cref{alg_two_agents_main_update_begin} and the monotonic property in the proof of \Cref{lem:two_agents_localsearch}, after executing \cref{alg_two_agents_main_update_ls}, the utility difference between the two bundles in~$\mathcal{A}^{3-i}$ under~$u_{3-i}$ decreases (to see this, \Cref{lem:two_agents_localsearch} implies that the utility difference between the two bundles in the updated $\mathcal{A}^{3-i}$ is weakly smaller than the utility difference of the two bundles in $\mathcal{A}^i$ under $u_{3-i}$, and the condition at \cref{alg_two_agents_main_update_begin} implies that the utility difference between the two bundles in~$\mathcal{A}^{3-i}$ strictly decreases after executing \cref{alg_two_agents_main_update_ls}).
We first notice that \cref{alg_two_agents_localsearch_swap1} cannot be executed in this call of the subroutine because of the failure of the conditions at \cref{alg_two_agents_main_ef_begin} or the previously executed \cref{alg_two_agents_main_ef2}.
If there is no swap step (\cref{alg_two_agents_localsearch_swap}) executed in this call of the subroutine, since the original allocation $(A^i_1,A^i_2)$ is EFX under $u_i$, either $u_i(A^{3-i}_1)\ge u_i(A^{3-i}_2)$ and then meets the condition at \cref{alg_two_agents_main_ef2}, or the allocation $\mathcal{A}^{3-i}$ is also EFX for agent~$i$.
In both cases, the algorithm terminates.
Thus, we then assume some swap step is executed in this call of subroutine, which means that when the item last transferred from one bundle to another under $u_{3-i}$ is $g$, the difference between the utilities of two bundles in $\mathcal{A}^{3-i}$ under $u_{3-i}$ is from a value larger than $u_{3-i}(g)$ to a value less than $u_{3-i}(g)$.

The difference for each $\mathcal{A}^i,i\in[2]$ under $u_i$ is decreasing during the whole process before terminating.
Similar to the analysis in Appendix~\ref{appendix:two_agents_localsearch}, each item that has been moved between the two bundles will not be further moved by the same agent.
Thus, the total number of move operations is bounded by $2m$.
The overall time complexity is~$O(m\log m)$, where we may first sort the items according to $u_1$ and $u_2$ respectively in $O(m\log m)$, and then perform the above operations in linear time.
For the execution of line \ref{alg_two_agents_main_doublyefx_cond}, we can use a heap to maintain the item with the smallest utility in $O(\log m)$ time per operation. Since the total number of executions of this line can also be bounded by the total number of move operations, which is $O(m)$, the overall time complexity for this line remains $O(m\log m)$.

\subsection{Proof of Theorem~\ref{thm:2agentsEFM}}

We analyze the procedure described before the statement of this theorem.
In the first case where a single integral allocation is given (\cref{alg_two_agents_main_ef_end,alg_two_agents_main_update_end}), this is an EF allocation and thus \exan EF and \expo EFM.
We then assume there exist two allocations~$\mathcal{A}^1$ and~$\mathcal{A}^2$ where $\mathcal{A}^i$ is EFX under~$u_i$ for each~$i \in [2]$.
If there exists an agent~$i \in [2]$ such that the good~$d$ is in bundle~$A^i_2$, from the \expo EFX property, we have $u_i(A^i_2) \geq u_i(A^i_1)$ and $u_i(A^i_2) - u_i(d) \le u_i(A^i_1)$, thus there exists an~$\alpha \in [0, 1]$ such that agent~$i$ treats these two bundles equally after transferring an $\alpha$-fraction of the divisible good~$d$.
Thus, let agent~$3-i$ pick her preferred bundle first and we achieve an EF allocation.

For the last case where the good~$d$ is in the bundle with smaller utility for both agents, since each integral allocation in the support is EFX and there is no divisible good in the better bundle, this can also imply EFM.

\section{Omitted Details in Section~\ref{sec:n-agent-mixed}}
\label{append:nagentmixed}

In this section, we provide the remaining details for the proof of \Cref{thm:prop_efm}.

\subsection{Finding the Minimal Unmatchable Group}
\label{append.nagentmixed1}

In this section, we define the notion of \emph{minimal unmatchable group} of agents and describe an algorithm for finding such a group.
These will be used in the next two sections.

A bipartite graph $G = (U, V, E)$ is constructed, where~$U$ denotes the set of agents, $V$ denotes the set of indivisible goods, and $(i, g) \in E$ if $u_i(g) = b$.
We use~$\Gamma(S)$ for $S \subseteq U$ to denote the set of neighbors of~$S$ in~$G$.
Given a bipartite graph and a matching, an alternating path and an augmenting path are defined as follows.

\begin{definition}[Alternating path in a bipartite graph]
Given a matching in a bipartite graph, an \emph{alternating path} is a path that begins with an unmatched agent or an unmatched good and in which the edges belong alternately to the matching and not to the matching.
\end{definition}

\begin{definition}[Augmenting path in a bipartite graph]
An \emph{augmenting path} is an alternating path starting from an unmatched agent (resp., good) and terminating at an unmatched good (resp., agent).
\end{definition}

\begin{definition}[Perfect matching in a bipartite graph]
Given a bipartite graph~$G$, a \emph{perfect matching} with respect to a set of agents~$U' \subseteq U$ and a set of goods~$V' \subseteq V$ is a matching that matches each agent in~$U'$ to a unique good in~$V'$.
\end{definition}

Based on the bipartite graph~$G$, we define the \emph{minimal unmatchable group} in Definition~\ref{def:minimal_unmatchable}.

\begin{definition}[Minimal unmatchable group]
\label{def:minimal_unmatchable}
A group of agents~$Z$ is said to be a \emph{minimal unmatchable group} if $G = (U, V, E)$ contains a perfect matching between $U \setminus Z$ and $V \setminus \Gamma(Z)$ and there is no subset $S \subseteq Z$ such that
\begin{itemize}
\item $S$ has a perfect matching in the bipartite graph, and
\item there is no edge between each agent $i\in Z\setminus S$ and each good matched to $S$.
\end{itemize}
\end{definition}

\begin{algorithm}[t]
\caption{Finding minimal unmatchable group}
\label{alg:find_unmatchable_group}
\DontPrintSemicolon

\KwIn{Agents~$N$, indivisible goods~$M$, and agents' utility functions.}

Construct a bipartite graph $G = (N, M, E)$ where $(i, g) \in E$ if $u_i(g) = b$.\;
Find a maximum matching on~$G$.\;
Let $Z$ be the set of all reachable agents from all unmatched agents through alternating paths.\;

\Return{The minimal unmatchable group~$Z$}
\end{algorithm}

Intuitively, for each agent belonging to the minimal unmatchable group $Z$, her value to each item except $\Gamma(Z)$ is $a$.
If there exists a perfect matching for all agents, $Z = \emptyset$.

A straightforward algorithm to find the minimal unmatchable group is to iteratively add the minimal set of agents that violates Hall's condition into the group.
We also provide an alternative approach in \Cref{alg:find_unmatchable_group}.

\subsection{Reduction to $m \leq 2n-2$}
\label{append.nagentmixed2}

Our first step to handle the general case is to allocate the indivisible goods so that there are at most~$2n-2$ indivisible goods that remain unallocated after this step, and the partial integral allocation satisfies envy-freeness.
The pseudocode is shown in \Cref{alg:reduce_indivisible_number}.

\begin{algorithm}[t]
\caption{Partially allocating indivisible goods}
\label{alg:reduce_indivisible_number}
\DontPrintSemicolon

\KwIn{Agents~$N$, indivisible goods~$M$, and agents' utility functions.}

$\forall i \in [n]$, $A_i' \gets \emptyset$\;
Let $T \gets \emptyset$ be the set of agents that do not have a perfect matching.\;
Let $Y \gets \emptyset$ be $\Gamma(T)$ among the unallocated goods. \;
$\cnt \gets 0$, $t \gets 1$\;

\While{$|M|\ge \cnt + n$\label{line:prop_while_condition}}{
	Find the minimal unmatchable group~$Z_t$ among $N \setminus T$ and $M$ according to \Cref{alg:find_unmatchable_group}.\;
	$T \gets T \cup Z_t$, $Y \gets Y \cup \Gamma(Z_t)$\;
	$M \gets M \setminus \Gamma(Z_t)$ and update the bipartite graph.\;
	\lIf{$|M|< \cnt + n$\label{line:prop_if_condition}}{
		\Break
	}
	Find a perfect matching between~$N \setminus T$ and~$M$, and assume each agent~$i \in N \setminus T$ is matched to~$g_{i'} \in M$.\;
	$\forall i \in N \setminus T$, $A_i \gets A_i \cup \{g_{i'}\}$\;
	Remove the matched items from~$M$.\;
	Update the bipartite graph.\;
	$\cnt \gets \cnt + |T|$, $t \gets t+1$\;
}
$t \gets t-1$\;

\ForEach{$i \in [n]$}{ \label{line:prop_alloc_cnt}
	Arbitrarily allocate~$t - |A_i|$ goods from~$M$ to~$A_i$, and remove the allocated goods from~$M$.\;
}

Let $V\leftarrow M$\;
\Return{Partial allocation~$\mathcal{A}'$, and unallocated indivisible goods~$Y \cup V$}
\end{algorithm}

In each iteration of the algorithm, we run \Cref{alg:find_unmatchable_group} to find the minimal unmatchable group~$Z$.
If $Z = \emptyset$, we directly find a perfect matching for all agents and allocate the goods accordingly.
Otherwise, $Z$ and~$\Gamma(Z)$ (i.e., the neighbors of~$Z$) are removed from the bipartite graph.
We find a perfect matching between the remaining agents and the remaining goods, which, by the definition of the minimal unmatchable group, is guaranteed to exist.
During the process, a variable~$\cnt$ is maintained to indicate how many additional indivisible goods should be allocated after each iteration to ensure that each bundle contains the same number of indivisible goods.
The termination condition of the \verb|while|-loop in \cref{line:prop_while_condition,line:prop_if_condition} ensures that the goods will not run out.
After the number of goods is not enough to guarantee the same size of each bundle for the next iteration, we terminate the process and obtain a partial allocation.

In the following part, we use $\mathcal{A}' = (A_1',\ldots, A_n')$ to denote this partial allocation, $T$ to denote the set of all the unmatchable agents, and $Y$ to denote the neighbors of $T$ that are removed from $M$ during the algorithm.

\begin{lemma}
$|Y \cup V| \leq 2n-2$.
\end{lemma}

\begin{proof}
For each minimal unmatchable group~$Z$, we have $Z > \Gamma(Z)$; otherwise, Hall's condition is satisfied, and there is a perfect matching between~$Z$ and~$\Gamma(Z)$.
Hence, $|Y| < |T| \leq n$.
Moreover, the number of goods allocated in \cref{line:prop_alloc_cnt} is equal to~$\cnt$.
Hence, the number of goods in~$V$ is less than~$n$ guaranteed by the termination condition of the \verb|while|-loop, as desired.
\end{proof}

\begin{lemma}
\label{lem:prop_efm_partial_ef}
\Cref{alg:reduce_indivisible_number} returns a partial integral allocation that is envy-free.
\end{lemma}

\begin{proof}
First, \Cref{alg:reduce_indivisible_number} ensures that each agent receives the same number of goods.
For each agent in~$N \setminus T$, envy-freeness is guaranteed as she only receives goods of value~$b$.
For each agent~$i \in T$, assume she is added to~$T$ at round~$t_i$, then each good she receives before round~$t_i$ has value~$b$ to her, and each good allocated to any agent after round~$t_i$ (including~$t_i$) has value~$a$ to her.
Hence, $u_i(A_i) = (t_i - 1) \cdot b + (t - t_i + 1) \cdot a$ while $u_i(A_j) \leq (t_i - 1) \cdot b + (t - t_i + 1) \cdot a$, meaning that agent~$i$ is envy-free in the partial allocation.
\end{proof}

We have already shown that the partial integral allocation~$\mathcal{A}'$ satisfies EF and thus PROP over goods~$M' = \bigcup_{i = 1}^n A'_i$.
We observe that if there is a randomized allocation $\mathcal{A}$ over $(M \setminus M') \cup D$ that satisfies \exan PROP and \expo EFM, then we can construct the allocation over $M$ where each agent $i$ receives bundle $A_i$ deterministically and the corresponding bundle in ${\mathcal{A}}_j$ with probability $p_j$ which satisfies \exan PROP and \expo EFM.
This reduces our problem to the case with at most $2n-2$ indivisible goods.

In the next section, we assume there are~$n$ agents, at most~$2n-2$ indivisible goods, and a single divisible good~$\{d\}$.

\subsection{Allocating Remaining Goods}
\label{append.nagentmixed3}

As mentioned, we have assumed $m \leq 2n-2$.
The case where $m \leq n$ has already been handled in \Cref{sect:mlen}.
Therefore, we assume $n < m \leq 2n-2$ below.

\paragraph{Step~1:}
We run \Cref{alg:find_unmatchable_group} on~$N$ and~$M$ to find the minimal unmatchable group.
Denote the group we find by~$T$ and their neighbors~$\Gamma(T)$ by~$Y$.
This guarantees the existence of a perfect matching between~$N \setminus T$ and~$M \setminus Y$.
We then construct a bipartite graph $G = (N, M, E)$ for the agents and the indivisible goods where $(i, g) \in E$ if $u_i(g) = b$, and find a random permutation $\pi$ of the agents.
We denote the first~$k$ agents in~$\pi$ by~$\pi[:k]$.

We adopt a similar idea as in Round-Robin.
As $n < m \leq 2n-2$, each agent of~$\pi[:m-n]$ will receive two goods, and the other agents will receive one good.
Meanwhile, we execute some extra operations to ensure the allocation satisfies \exan PROP.
Next, we describe how to allocate the indivisible goods for agents in~$N \setminus T$ (resp.,~$T$) in Step~2 (resp., Step~3).

\paragraph{Step~2:}
First, we decide the allocation of indivisible goods for agents in~$N \setminus T$.
Consider the induced subgraph $G'=(N\setminus T, M\setminus Y, E')$.

\emph{First phase:}
Find a perfect one-to-one matching in~$G'$ and allocate the goods accordingly.
Note that the matching is fixed regardless of $\pi$.

\emph{Second phase:} Following the order in $\pi$, iteratively let each agent of $\pi[:m-n] \cap (N \setminus T)$ receive an unallocated good with the highest value.

\paragraph{Step~3:}
Next, we decide the allocation of indivisible goods for agents in~$T$.
Each agent receives exactly one good in the first phase and the remaining goods are handled in the second phase.

\emph{First phase:}
We handle the agents in~$T$ sequentially according to their order in~$\pi$.
Similar to Step~2, let each agent receive a good of value~$b$ if there is still such a good available and add this edge to the matching.
Otherwise, find an augmenting path and update the goods along the path as well as the matching.
If there is no augmenting path, we skip the current agent and let her receive an arbitrary good after finishing executing the aforementioned process to all of the agents in~$T$.

\emph{Second phase:}
After each agent in~$T$ receives a good, we further allocate an unallocated indivisible good of $M$ to each of $\pi[:m-n] \cap T$.

\paragraph{Step~4:}
Execute the water-filling process for divisible goods and obtain the allocation~$\mathcal{A}^\pi$.
The randomized allocation we obtain is $\{(\frac{1}{n!}, \mathcal{A}^{\pi})\}$ and is denoted by~$\mathcal{R}$.

\medskip
In what follows, we show that the randomized allocation~$\mathcal{R}$ is \exan PROP, and consider the agents in~$T$ and in~$N \setminus T$ in \Cref{lem:prop_for_x,lem:prop_for_n_setminus_x} separately and respectively.
We start with a \namecref{lem:step_4} which will be used in the proof of \Cref{lem:prop_for_x}.

\begin{lemma}
\label{lem:step_4}
The indivisible goods in~$Y$ are fully allocated after the first phase in Step~3.
\end{lemma}

\begin{proof}
Consider the induced subgraph $G^* = (T, Y, E^*)$.
Assume there exists a good~$g \in Y$ being unallocated after the first phase.
We collect all the alternating paths starting from~$g$ and denote the set of vertices belonging to~$T$ on the paths as~$T_1$, where $T_1 \neq \emptyset$ since~$g \in \Gamma(T)$.
We argue that each agent in~$T_1$ is matched in~$G^*$; otherwise, there is an augmenting path, and the allocation should have been updated during the first phase of Step~3.
Every agent in~$T \setminus T_1$ values each good matched to agents~$T_1$ at~$a$ as we have already found all the alternating paths starting from~$g$.
Let $S \subsetneq T_1$ and $S \neq \emptyset$.
Then, $S$ is an evidence demonstrating that~$T$ is not the minimal unmatchable group (see \Cref{def:minimal_unmatchable}), a contradiction.
\end{proof}

\begin{lemma}
\label{lem:prop_for_x}
The randomized allocation~$\mathcal{R}$ is ex-ante PROP for agents in~$T$.
\end{lemma}

\begin{proof}
We define baseline allocation as in \Cref{def:baseline_allocation}, and adopt a similar analysis as in \Cref{thm:prop_efm_m_le_n}.
For any agent~$i \in T$, as there is no edge from~$T$ to~$M \setminus Y$ and $|Y| < |T| \leq n$, the number of goods valued at~$b$ by agent~$i$ is at most~$n$.
Hence, each bundle~$B^i_k$ where~$k \in [n]$ in the baseline allocation~$\mathcal{B}^i$ contains at most one indivisible good of value~$b$ and the number of bundles that contain two indivisible goods with values~$b$ and~$a$ respectively is maximized according to the property of Round-Robin.

We use $A(\pi, k)$ to denote the bundle allocated to the agent ranked $k$-th in the permutation~$\pi$ in the allocation output by our algorithm.
Consider an arbitrary permutation~$\pi$ where agent~$i$ is ranked $k$-th (i.e., $\pi(k) = i$), we will show that $u_i(A(\pi, k)) \geq u_i(B_k^i)$.
Since an agent in $\pi[:m-n]$ receives two goods and otherwise one good in our algorithm, the number of indivisible goods received by agent~$i$ is the same as that in the baseline allocation before the water-filling process with whichever index in~$\pi$.
In addition, she will receive a good with value~$b$ if $B^i_k$ contains one good.
Hence, $u_i(A(\pi, k)) \geq u_i(B_k^i)$ if $B_k^i$ contains only indivisible goods.

If $B^i_k$ contains a fraction of the divisible good~$D$, let $x = u_i(B^i_k)$.
Denote the actual output allocation by $\mathcal{A} = (A_1, \dots, A_n)$, and denote the set of bundles that only contain indivisible goods in the actual allocation~$\mathcal{A}$ by $\mathcal{S}\subseteq \mathcal{A}$.
By \Cref{lem:step_4} and according to the first phase of Step~3, each agent in~$T$ receives at most one good from~$Y$, so no bundle in the actual allocation contains two indivisible goods both valued at~$b$ to agent~$i$.
As the number of bundles with the maximum possible value of indivisible goods (which is $b+a$) is maximized in~$\mathcal{B}^i$, we have $\sum_{A_j\in \mathcal{S}}u_i(A_j)\le \sum_{j=1}^{|\mathcal{S}|}u_i(B_j^i)$, which holds even if we do not consider all the divisible goods in $\bigcup_{j=1}^{|\mathcal{S}|}B_j^i$.
Therefore, if $u_i(A_i)=x' < x$ in the actual allocation~$\mathcal{A}$, we will reach the following contradiction
\begin{align*}
u_i(M \cup D)
&=\sum_{A_j\in\mathcal{S}}u_i(A_j)+\sum_{A_j\notin\mathcal{S}}u_i(A_j)
\le \sum_{A_j\in\mathcal{S}}u_i(A_j)+(n-|\mathcal{S}|)\cdot x' \\
&< \sum_{j=1}^{|\mathcal{S}|}u_i(B_j^i) + (n-|\mathcal{S}|)\cdot x
\le \sum_{j=1}^{|\mathcal{S}|}u_i(B_j^i) + \sum_{j=|\mathcal{S}|+1}^{n}u_i(B_j^i) = u_i(M \cup D),
\end{align*}
where the first inequality holds as each bundle $A_j\notin \mathcal{S}$ that contains divisible goods is valued at most $x'$ to agent $i$ guaranteed by EFM, and the third inequality holds as $u_i(B_j^i)\ge x$ if $B_j^i$ contains no divisible good.

As the permutation is generated uniformly at random, the lemma concludes following the same analysis for \exan PROP as in \Cref{thm:prop_efm_m_le_n}:
\[
u_i(\mathcal{R}) = \sum_{k = 1}^n \sum_{\pi : \pi(k) = i} \frac{u_i(A(\pi, k))}{n!} \geq \sum_{k = 1}^n \sum_{\pi : \pi(k) = i} \frac{u_i(B_k^i)}{n!} = \sum_{k = 1}^n \frac{u_i(B^i_k)}{n} = \frac{u_i(M \cup D)}{n}.
\qedhere
\]
\end{proof}

\begin{lemma}
\label{lem:prop_for_n_setminus_x}
The randomized allocation~$\mathcal{R}$ is ex-ante PROP for agents in~$N \setminus T$.
\end{lemma}

\begin{proof}
Different from the proof in \Cref{lem:prop_for_x} above, a bundle in the actual allocation~$\mathcal{A}$ may contain two indivisible goods both valued at~$b$ to agent~$i$ for some~$i \in N \setminus T$.
We thus define the baseline allocation for agent~$i$ in a different way below.

\begin{definition}[Baseline allocation~$\mathcal{B}^i$ for agent~$i \in N \setminus T$]
\label{def:baseline_allocation_N_setminus_x}
The \emph{baseline allocation~$\mathcal{B}^i$ for agent~$i \in N \setminus T$} is obtained by
\begin{enumerate}
\item letting each bundle contain one distinct good that is matched during the first phases of Steps~2 and~3 (note that there are~$n$ goods matched during the first phases of Steps~2 and~3, so each bundle can contain exactly one good);
\item arranging the bundles in descending order according to agent~$i$'s utility;
\item letting agent~$i$ allocate remaining indivisible goods by Round-Robin; and
\item executing the water-filling process for the divisible goods according to agent~$i$'s utility.
\end{enumerate}
\end{definition}

Assume that agent~$i$ is ranked $k$-th in the permutation~$\pi$, and denote the actual output allocation by $\mathcal{A}=\{A_1,\ldots,A_n\}$, where $A_i$ is also denoted by $A(\pi,k)$ similar to Lemma~\ref{lem:prop_for_x}.
The number of the indivisible goods in $B_k^i$ equals to that she receives in the actual allocation before the water-filling process.
In the first phase of Step~2, agent $i$ will receive a good with value $b$, which is no less than the value of the first good added to $B_k^i$ (as the value is at most $b$).
If the second good added to $B_k^i$ is also valued at $b$ to agent $i$, then the good she receives in the second phase of Step~2 also has value $b$ by Round-Robin.
Therefore, if $B_k^i$ contains only indivisible goods, the property $u_i(A_i)\ge u_i(B_k^i)$ holds, as agent $i$'s value to her bundle only increases during the water-filling process.

We provide a further observation of the indivisible goods in the two allocations before the water-filling process.
Denoted by $\bar{A}_j$ and $\bar{B}^i_j$ the $j$-th bundles that contain only the indivisible goods before the water-filling process for $j\le n$.
The first good added to $\bar{A}_j$ and $\bar{B}^i_j$ is identical.
Guaranteed by the second and third steps when constructing the baseline allocation, the possible number of bundles with two goods both with value $b$ to agent $i$ is maximized, i.e., there is no actual allocation such that the bundles with two goods with value $b$ is strictly larger than that in the baseline allocation.
Similarly, subject to the above property, the possible number of bundles with two goods with values $a$ and $b$ respectively is maximized.
This allows us to obtain that for any $k'\le n$, $\sum_{j=1}^{k'}u_i(\bar{B}^i_j)\ge \sum_{\bar{A}_j\in\mathcal{S}}u_i(\bar{A}_j)$ for any set $\mathcal{S}$ of bundles where $|\mathcal{S}|=k'$.

We similarly denote $x=u_i(B_k^i)$ if $B_k^i$ contains a fraction of the divisible good.
Denote the set of bundles without any divisible good by $\mathcal{S}$.
From the above analysis, we have $\sum_{j=1}^{k'}u_i(B^i_j)\ge \sum_{\bar{A}_j\in\mathcal{S}}u_i(A_j)$.
Then, if $u_i(A_i)=x'< x$, we similarly obtain the contradiction that
\begin{align*}
u_i(M \cup D)
&=\sum_{A_j\in\mathcal{S}}u_i(A_j)+\sum_{A_j\notin\mathcal{S}}u_i(A_j)
\le \sum_{A_j\in\mathcal{S}}u_i(A_j)+(n-|\mathcal{S}|)\cdot x' \\
&< \sum_{j=1}^{|\mathcal{S}|}u_i(B_j^i) + (n-|\mathcal{S}|)\cdot x
\le \sum_{j=1}^{|\mathcal{S}|}u_i(B_j^i) + \sum_{j=|\mathcal{S}|+1}^{n}u_i(B_j^i) = u_i(M \cup D),
\end{align*}

We have shown that $u_i(A(\pi,k))\geq u_i(B_k^i)$ holds for each $\pi$ and $i$ such that $\pi(k)=i$.
Therefore, for each $\pi$, each agent will receive no less value in the actual allocation than in the baseline allocation.
Following the same analysis in \Cref{lem:prop_for_x}, \exan PROP is guaranteed.
\end{proof}

We are now ready to establish \Cref{thm:prop_efm}.

\begin{proof}[Proof of \Cref{thm:prop_efm}]
\Cref{lem:prop_for_x,lem:prop_for_n_setminus_x} together show the randomized allocation is \exan PROP, and it is also \expo EFM guaranteed by the fact that each partial allocation containing only indivisible goods is EF1 and by Lemma~\ref{lem:efm}.
Finally, it is straightforward that each part of our algorithm runs in polynomial time.
This concludes Theorem~\ref{thm:prop_efm}.
\end{proof}

\begin{remark}
The above algorithm can be generalized to the case where agents have personalized bi-valued utilities (the value $a$ and $b$ of each agent are different) to indivisible goods and no constraint to divisible goods.
\end{remark}

\section{Omitted Details in Section~\ref{sec:n-agent-indivisible}}

\subsection{Proof of \Cref{lem:efx_fpo_expost_efx}}

We can verify the required property after each round and prove this by induction.
The initial empty allocation is trivially EFX.
(Recall that we use~$N^r$ and~$M^r$ to denote the corresponding set~$N$ and~$M$ encountered in \Cref{alg:efx_fpo} at the beginning of its round~$r$.)
At the next round~$r$, if we cannot find a perfect matching between~$N^r$ and goods~$M^r$ (lines~\ref{alg_efx_fpo_case2_begin}-\ref{alg_efx_fpo_case2_end}), there exists a non-empty minimal unmatchable group~$Z_t$.
For each agent~$i \in N^r \setminus Z_t$ at this round, since she can still receive a large good and every agent can receive at most one good per round, there is no envy from agent~$i$.

For each agent~$i \in Z_t$, if she is frozen in this loop, since she will be only frozen for the next $\lfloor b/a-1 \rfloor$ rounds, the total utility accumulated by any other agent's bundle during these frozen rounds is upper bounded by $a + \lfloor b/a-1 \rfloor \cdot a \leq b$, from the second term in Observation~\ref{obs:efx_fpo}.
Thus, each frozen agent will not envy others during the frozen rounds.
If an agent~$i \in Z_t$ is not frozen (in other words, becomes \emph{quiet}), she will not envy the agents not in $Z_t$ or the agents in $Z_t$ without freezing since these agents can only receive a good in $M^r \setminus \Gamma(Z_t)$, which is valued at~$a$ under~$u_i$.
For some agent~$j \in Z_t$ is frozen in this round, if all goods allocated to~$j$ until now are valued at~$b$ by agent~$i$, agent~$i$ will not envy $j$ after removing the last good from the second term in Observation~\ref{obs:efx_fpo}.
Otherwise, agent $i$ will not envy $j$ since there is exactly one small good in $i$'s bundle and there is at least one small good in $j$'s bundle under $u_i$.

For each agent $i$ not in $N^r$, i.e., a quiet agent in a previous $Z_{t'}$, it does not envy the agents not in $Z_{t'}$ and the agents in $Z_{t'}$ without freezing before because all agents receive a small good from $i$'s perspective.
    Compared to some other quiet agent $j$ in $Z_{t'}$ who has been frozen before, after eliminating one small good in $j$'s bundle, we have $b\le a+ \lfloor b/a-1 \rfloor \cdot a +a$, which can eliminate the envy. The term $a+ \lfloor b/a-1 \rfloor \cdot a$ represents the utility during the frozen rounds of $j$ earned by $i$ and the last $a$ represents the utility at this round.

    We then come to the second case where we can find a perfect matching between $N^r$ and goods $M^r$ (Lines~\ref{alg_efx_fpo_case1_begin}-\ref{alg_efx_fpo_case1_end}). For each agent $i$ in $N^r$, since it receives large goods at all rounds, there is no envy from her.
    For each agent $i$ in some previous $Z_t$, the case is similar to the case above.

    From the induction, it suffices to show we can still maintain EFX when allocating the remaining goods at Lines \ref{alg_efx_fpo_rem_begin}-\ref{alg_efx_fpo_rem_end}.
    Because every agent $i$ does not envy agent $j$ except for one special case as stated in the property we maintain, the envy between these agents can be eliminated by removing the good allocated at this final step. The remaining is to show the case where $i$ and $j$ are in the same $Z_t$ and agent $i$ has never been frozen but agent $j$ has been frozen before.
    If the value of agent $i$ over the good $g$ allocated to agent $j$ at round $r_j$ is $a$, the envy can be eliminated after removing the good allocated at this final step since there is no envy from $i$ to $j$ during the loop.
    Otherwise, $j$ must be before $i$ in the reversed order of $\pi$. This is because if $i$ is before $j$ in the reversed order of $\pi$, there must exist an augmenting path from $i$ to the good $g$ (from the fact that $i$ can directly take the good $g$ and we keep the remaining allocation to reach a matching with a larger size before the position of $j$) and then $i$ would be frozen.
    Thus, $i$ will pick a good before $j$ at this final step from $\pi$, which means the envy from $i$ to $j$ can be eliminated after removing the good received by $j$ at this final step (if it exists).

\subsection{Proof of Lemma~\ref{lem:efx_fpo_exante_ef}}

For each pair of two agents $i$ and $j$ except for the case that agent $i$ and agent $j$ are in the same $Z_t$ and agent $i$ has never been frozen while agent $j$ has been frozen, from the maintained property as stated before Lemma~\ref{lem:efx_fpo_expost_efx}, it suffices to show the allocation at the final step (Lines \ref{alg_efx_fpo_rem_begin}-\ref{alg_efx_fpo_rem_end}) can ensure \exan EF.
If agent $i$ envies agent $j$ in some allocation $\mathcal{A^{\pi}}$, either (1) $j$ receives a good $g$ but $i$ does not at the final step or (2) $j$ receives a good $g$ valued $b$ by $i$ but $i$ only receives a small good at the final step.
If we compare it to the allocation under $\pi'$ which only exchanges $i$ and $j$ in $\pi$, the envy in the first case can be eliminated because at the final step, $j$ will get nothing and $i$ can receive a good with a weakly larger utility than $u_i(g)$ through an augmenting path, where the worst case is that agent $i$ directly receives $g$ and the remaining allocation is kept.

For the second case (this implies that $j$ is before $i$ in $\pi$ and $j$ also values $g$ at $b$), the only thing we need to show is that agent $i$ and $j$ cannot both get a large good from $i$'s perspective at the final step under $\pi'$.
Here, we assume $i$ is at position $p^{(i)}$ and $j$ is at position $p^{(j)}$ in $\pi$. We have $p^{(i)}>p^{(j)}$.
If the above situation happens and assume $i$ and $j$ receive $g^{(i)}$ and $g^{(j)}$ under $\pi'$, we can find an augmenting path to let agent $i$ receive a large good under $\pi$, which leads to a contradiction. The way to find such a path is as follows.

Since there is a matching for $\pi'$ by replacing the matching edge $(j,g)$ by $(i,g)$ in $\pi$, the size of the maximum matching before position $p^{(i)}$ in $\pi$ is upper bounded by the size of that before $p^{(i)}$ in $\pi'$.
If $j$ values $g^{(j)}$ at $b$, this violates the maximality of the matching up to position $p^{(i)}$ in $\pi$ since the size should be equal in both maximum matchings up to position $p^{(i)}$ in $\pi$ and $\pi'$.
If $j$ values $g^{(j)}$ at $a$, we then can get the size of the maximum matching before position $p^{(i)}$ in $\pi$ is exactly equal to the size of that before position $p^{(i)}$ in $\pi'$ and $j$ values $g^{(i)}$ at $b$. We then can replace the allocation before position $p^{(i)}$ in $\pi$ by that before position $p^{(i)}$ in $\pi'$ and then agent $i$ can receive $g''$ in $\pi$, which violates the maximality of the matching up to position $p^{(i)}$ in $\pi$. Both cases lead to a contradiction, which finishes the proof for this case.

We then come to the exceptional case: agent $i$ and $j$ are in the same $Z_t$ and agent $i$ has never been frozen but agent $j$ has been frozen.
If the value of agent $i$ over the good $g$ allocated to $j$ at round $r_j$ is $a$, there is no envy before the final step and this case can be solved similarly as in the above analysis.

Otherwise, we assume the generated permutation is $\pi$.
We consider another permutation $\pi'$ which only exchanges the positions of $i$ and $j$ in $\pi$.
Since the \exan guarantee at the final step has been shown in the above situation, it suffices to show the \exan EF is also satisfied before the final step.
Because agent $j$ is frozen under $\pi$ and agent $i$ values the good $g$ at $b$, agent $i$ will be frozen under $\pi'$ by directly replacing the agent $j$ in $\pi$.
If agent $i$ is frozen and agent $j$ is not frozen under $\pi'$, then from $i$'s perspective, the differences between the bundles of agent $i$ and $j$ under $\pi$ and $\pi'$ are the same, which ensures the \exan EF.
It suffices to show that agents $i$ and $j$ cannot be frozen simultaneously under $\pi'$.

If both agents $i$ and $j$ are frozen under $\pi'$, we can utilize the similar argument over the maximality of the matching before (or up to) the position of $i$ in $\pi$ to show that $i$ can be also frozen under $\pi$, which violates the assumption. This completes the proof.

\subsection{Descriptions of Fisher Market}

In a Fisher market, we are given a set of $n$ agents denoted by $N$ and a set of $m$ \emph{divisible} goods denoted by $M$.
For each agent $i\in N$, an additive utility function $u_i$  and a budget $e_i\ge 0$ are given.
Given a fractional allocation $\X$ with a price vector $\bf{p}$ over $M$, the spending of an agent $i$ under $(\X,\bf{p})$ is $\sum_{g\in M}p_g X_{ig}$.
We then define the \emph{maximum bang-per-buck (MBB) ratio} $\alpha_i \coloneqq \max_{g\in M} u_i(g)/p_g$ and the \emph{MBB-set} $\rm{MBB}_i \coloneqq \{g \in M \colon u_i(g)/p_g = \alpha_i\}$ for each agent $i\in N$.

A \emph{market equilibrium} $(\X,\bf{p})$ is a fractional allocation $\X$ with a price vector $\bf{p}$ over $M$ which satisfies:
\begin{itemize}
    \item the market clears, i.e. $\sum_{i\in N} X_{ig}=1$ for each good $g\in M$;
    \item the budget is fully spent, i.e., $\sum_{g\in M}p_g X_{ig}=e_i$ for each agent $i\in N$;
    \item each agent $i\in N$ only receives goods in $\rm{MBB}_i$.
\end{itemize}

\begin{theorem}[First Welfare Theorem~\citep{mas1995microeconomic}]
\label{thm:fisher_market}
    If $(\X,\bf{p})$ is a market equilibrium of a Fisher market, $\X$ is fPO.
\end{theorem}

\subsection{Proof of Lemma~\ref{lem:efx_fpo_expost_fpo} when $m<n$}

    We now come to the corner case when $m<n$. The main difference in this case is the price setting for low-price agents in the above analysis, where it is possible that there exists a low-price agent receiving only a small good. Thus, we need to specify the price separately for this case.

    We define the price vector $\bf{p}$ in the following way:
    for each agent $i$ receiving a small good $g$, we set $p_g=a$. We then set the prices of the remaining goods iteratively: if there exists one agent $i$ without pricing her good $g$ such that some priced agent values $g$ at $b$, we set $p_g=b$.
    If there is no further pricing can be made, we set the price as $a$ for each remaining good and call each agent who owns these goods low-price agent, correspondingly.
    We can then follow a similar proof as the case above to prove this.
\end{document}